\documentclass[
aps,
pra,
reprint,
superscriptaddress
]{revtex4-2}
\usepackage[T1]{fontenc}
\usepackage[utf8]{inputenc}
\usepackage{lmodern}
\usepackage{mathtools}
\usepackage{amssymb}
\usepackage{amsthm}
\usepackage{bm}
\usepackage{graphicx}
\usepackage{dcolumn}
\usepackage{xcolor}
\usepackage[colorlinks=true,
citecolor=blue,
linktocpage=true
]{hyperref}
\usepackage[capitalise]{cleveref}
\usepackage{booktabs}

\newcommand{\nocontentsline}[3]{}
\newcommand{\tocless}[2]{\bgroup\let\addcontentsline=\nocontentsline#1{#2}\egroup}

\newtheorem{theorem}{Theorem}
\newtheorem{lemma}[theorem]{Lemma}
\newtheorem{corollary}[theorem]{Corollary}
\newtheorem{proposition}[theorem]{Proposition}

\newtheorem{remark}[theorem]{Remark}
\newtheorem{definition}[theorem]{Definition}

\renewcommand{\O}{\mathcal{O}}

\newcommand{\R}{\mathbb{R}}
\newcommand{\C}{\mathbb{C}}

\newcommand{\Z}{\mathbb{Z}}

\newcommand{\MA}[2]{\mathcal{A}_{#2}^{(#1)}}

\newcommand{\SO}{\mathrm{SO}}
\newcommand{\SU}{\mathrm{SU}}
\newcommand{\U}{\mathrm{U}}
\newcommand{\FGU}{\mathrm{FGU}}
\newcommand{\NC}{\mathrm{NC}}
\newcommand{\Cl}{\mathrm{Cl}}
\newcommand{\Sym}{\mathrm{Sym}}
\newcommand{\Alt}{\mathrm{Alt}}
\newcommand{\Stab}{\mathrm{Stab}}
\DeclareMathOperator{\spn}{span}
\DeclareMathOperator{\E}{\mathbb{E}}
\DeclareMathOperator{\V}{Var}

\DeclareMathOperator{\tr}{tr}

\DeclareMathOperator{\loc}{loc}
\newcommand{\comb}[2]{\mathcal{C}_{#1, #2}}
\newcommand{\diags}[2]{\mathcal{D}_{#1, #2}}
\newcommand{\op}[2]{\ket{#1}\!\bra{#2}}

\newcommand{\ev}[3]{\langle #1 | #2 | #3 \rangle}
\newcommand{\ket}[1]{| #1 \rangle}
\newcommand{\bra}[1]{\langle #1 |}
\newcommand{\T}{\mathsf{T}}

\newcommand{\sns}[2]{\| #1 \|_{#2}}
\renewcommand{\l}[1]{\mathopen{}\left#1}
\renewcommand{\r}[1]{\right#1\mathclose{}}

\newcommand{\fs}{f_{\mathrm{samp}}}
\newcommand{\tl}{t_{\mathrm{load}}}

\begin{document}

\title{Fermionic partial tomography via classical shadows}

\author{Andrew Zhao}
\email{azhao@unm.edu}
\affiliation{Center for Quantum Information and Control, Department of Physics and Astronomy, University of New Mexico, Albuquerque, New Mexico 87106, USA}

\author{Nicholas C. Rubin}
\email{nickrubin@google.com}
\affiliation{Google Research, Mountain View, California, 94043, USA}

\author{Akimasa Miyake}
\email{amiyake@unm.edu}
\affiliation{Center for Quantum Information and Control, Department of Physics and Astronomy, University of New Mexico, Albuquerque, New Mexico 87106, USA}

\date{\today}

\begin{abstract}
	We propose a tomographic protocol for estimating any $ k $-body reduced density matrix ($ k $-RDM) of an $ n $-mode fermionic state, a ubiquitous step in near-term quantum algorithms for simulating many-body physics, chemistry, and materials. Our approach extends the framework of classical shadows, a randomized approach to learning a collection of quantum-state properties, to the fermionic setting. Our sampling protocol uses randomized measurement settings generated by a discrete group of fermionic Gaussian unitaries, implementable with linear-depth circuits. We prove that estimating all $ k $-RDM elements to additive precision $ \varepsilon $ requires on the order of $ \binom{n}{k} k^{3/2} \log(n) / \varepsilon^2 $ repeated state preparations, which is optimal up to the logarithmic factor. Furthermore, numerical calculations show that our protocol offers a substantial improvement in constant overheads for $ k \geq 2 $, as compared to prior deterministic strategies. We also adapt our method to particle-number symmetry, wherein the additional circuit depth may be halved at the cost of roughly 2--5 times more repetitions.
\end{abstract}

\maketitle

\emph{Introduction.}---One of the most promising applications of quantum computation is the study of strongly correlated systems such as interacting fermions. While quantum algorithms such as phase estimation~\cite{nielsen2002quantum,kitaev2002classical} allow for directly computing important quantities such as ground-state energies with quantum speedup~\cite{abrams1999quantum,somma2002simulating,aspuru2005simulated}, current hardware limitations~\cite{preskill2018quantum} have directed much attention toward variational methods. Of note is the variational quantum eigensolver (VQE)~\cite{peruzzo2014variational,mcclean2016theory}, where short-depth quantum circuits are repeatedly executed in order to estimate observable expectation values.

Initial bounds on the number of these circuit repetitions associated with fermionic two-body Hamiltonians were prohibitively high~\cite{wecker2015progress}, spurring on much recent work addressing this problem. We roughly classify these strategies into two categories:~those that specifically target energy estimates~\cite{mcclean2014exploiting,mcclean2016theory,kandala2017hardware,babbush2018low,rubin2018application,izmaylov2019revising,izmaylov2019unitary,huggins2019efficient,crawford2019efficient,zhao2020measurement,torlai2020precise,arrasmith2020operator,paini2019approximate,hadfield2020measurements,yen2020cartan,gonthier2020identifying,huang2021efficient,garcia2021learning,hillmich2021decision,hadfield2021adaptive,wu2021overlapped}, referred to as Hamiltonian averaging, and more general techniques that can learn the $ k $-body reduced density matrices ($ k $-RDMs) of a quantum state~\cite{aaronson2020shadow,*aaronson2018online,*aaronson2019gentle,yu2019quantum,*yu2020sample,verteletskyi2020measurement,jena2019pauli,yen2020measuring,gokhale2019minimizing,*gokhale2019on3,cotler2020quantum,bonet2019nearly,hamamura2020efficient,garcia2020pairwise,jiang2020optimal,evans2019scalable,huang2020predicting,smart2020lowering,tilly2021reduced}. (Not all works fit neatly into this dichotomy, e.g., Refs.~\cite{harrow2019low,wang2019accelerated,kubler2020adaptive,sweke2020stochastic,van2020measurement,wang2021minimizing}.) Hamiltonian averaging is ultimately interested in a single observable, allowing for heavy exploitation in its structure. In contrast, reconstructing an RDM requires estimating all the observables that parametrize it.

Though generally more expensive than Hamiltonian averaging, calculating the $ k $-RDM allows one to determine the expectation value of any $ k $-body observable~\cite{coleman1980reduced}. For example, the electronic energy of chemical systems is a linear functional of the 2-RDM, while in condensed-matter systems, effective models for electrons can require knowledge of the 3-RDM~\cite{tsuneyuki2008transcorrelated,PhysRevB.87.245129}. Beyond the energy, other important physical properties include pair-correlation functions and various order parameters~\cite{mazziotti2012two,jensen2017introduction}. The 2-RDM is also required for a host of error-mitigation techniques for near-term quantum algorithms~\cite{mcclean2017hybrid,rubin2018application,takeshita2020increasing}, which have been experimentally demonstrated to be crucial in obtaining accurate results~\cite{colless2018computation,sagastizabal2019experimental,mccaskey2019quantum,arute2020hartree}. Additionally, promising extensions to VQE such as adaptive ansatz construction~\cite{grimsley2019adaptive,ryabinkin2020iterative,tang2019qubit,wang2020resource} and multireference- and excited-state calculations~\cite{mcclean2017hybrid,parrish2019quantum,takeshita2020increasing,huggins2020non,stair2020multireference,urbanek2020chemistry} can require up to the 4-RDM.

Motivated by these considerations, in this work we focus on partial tomography for fermionic RDMs. While numerous works have demonstrated essentially optimal sample complexity for estimating qubit RDMs~\cite{cotler2020quantum,bonet2019nearly,jiang2020optimal,evans2019scalable,huang2020predicting}, such approaches necessarily underperform in the fermionic setting. Recognizing this fundamental distinction, Bonet-Monroig \emph{et al.}~\cite{bonet2019nearly} and Jiang \emph{et al.}~\cite{jiang2020optimal} developed measurement schemes that achieve optimal scaling for fermions. However, the former construction is not readily generalizable for $ k > 2 $, while the latter requires a doubling in the number of qubits and a specific choice of fermion-to-qubit mapping.

In this Letter, we propose a randomized scheme that is free from these obstacles. It is based on the theory of classical shadows~\cite{huang2020predicting}:~a protocol of randomly distributed measurements from which one acquires a partial classical representation of an unknown quantum state (its ``shadow''). Classical shadows are sufficient for learning a limited collection of observables, making this framework ideal for partial state tomography. Our key results identify efficient choices for the ensemble of random measurements, suitable for the structure of fermionic RDMs.

\emph{Fermionic RDMs.}---Consider a fixed-particle state $ \rho $ represented in second quantization on $ n $ fermion modes. The $ k $-RDM of $ \rho $, obtained by tracing out all but $ k $ particles, is typically represented as a $ 2k $-index tensor,
\begin{equation}
	{}^k D_{q_1 \cdots q_k}^{p_1 \cdots p_k} \coloneqq \tr(a_{p_1}^\dagger \cdots a_{p_k}^\dagger a_{q_k} \cdots a_{q_1} \rho),
\end{equation}
where $ a_p^\dagger, a_p $ are fermionic creation and annihilation operators, $ p \in \{0, \ldots, n-1\} $.
By linearity, these matrix elements may be equivalently expressed using Majorana operators, starting with the definitions
\begin{equation}
	\gamma_{2p} \coloneqq a_p + a_p^\dagger, \quad \gamma_{2p+1} \coloneqq -i (a_p - a_p^\dagger).
\end{equation}
Then for each $ 2k $-combination $ \bm{\mu} \equiv (\mu_1, \ldots, \mu_{2k}) $, where $ 0 \leq \mu_1 < \cdots < \mu_{2k} \leq 2n-1 $, we define a $ 2k $-degree Majorana operator
\begin{equation}\label{eq:majorana_defn}
	\Gamma_{\bm{\mu}} \coloneqq (-i)^{k} \, \gamma_{\mu_1} \cdots \gamma_{\mu_{2k}}.
\end{equation}
All unique $ 2k $-degree Majorana operators are indexed by the set of all $ 2k $-combinations of $ \{ 0, \ldots, 2n-1 \} $, which we shall denote by $ \comb{2n}{2k} $. Because Majorana operators possess the same algebraic properties as Pauli operators (Hermitian, self-inverse, and Hilbert--Schmidt orthogonal), any fermion-to-qubit encoding maps between the two in a one-to-one correspondence.

The commutativity structure inherited onto $ \comb{2n}{2k} $ constrains the maximum number of mutually commuting (hence simultaneously measurable) operators to be $ \O(n^k) $~\cite{bonet2019nearly}. As there are $ \O(n^{2k}) $ independent $ k $-RDM elements, this implies an optimal scaling of $ \O(n^k) $ measurement settings to account for all matrix elements.

\emph{Classical shadows and randomized measurements.}---We briefly review the framework of classical shadows introduced by Huang \emph{et al.}~\cite{huang2020predicting}, upon which we build our fermionic extension and prove sampling bounds. Let $ \rho $ be an $ n $-qubit state and $ \{ O_1, \ldots, O_L \} $ a set of $L$ traceless observables for which we wish to learn $ \tr(O_1 \rho), \ldots, \tr(O_L \rho)$. Classical shadows require a simple measurement primitive:~for each preparation of $ \rho $, apply the unitary map $ \rho \mapsto U \rho U^\dagger $, where $ U $ is randomly drawn from some ensemble $ \mathcal{U} $;~then perform a projective measurement in the computational basis, $ \{ \ket{z} \mid z \in \{ 0,1 \}^n \} $.

Suppose we have an efficient classical representation for inverting the unitary map on postmeasurement states, yielding $ U^\dagger \op{z}{z} U $. Then the process of repeatedly applying the measurement primitive and classically inverting the unitary may be viewed, in expectation, as the quantum channel
\begin{equation}
	\mathcal{M}_{\mathcal{U}}(\rho) \coloneqq \E_{U\sim\mathcal{U}, \ket{z} \sim U \rho U^\dagger} \l[ U^\dagger \op{z}{z} U \r],
\end{equation}
where $ \ket{z} \sim U \rho U^\dagger $ is defined by the usual probability distribution from Born's rule, $ \Pr[\ket{z} \mid U \rho U^\dagger] = \ev{z}{U \rho U^\dagger }{z} $. Informational completeness of $\mathcal{U}$  ensures that this channel is invertible, which allows us to define the classical shadow
\begin{equation}
	\hat{\rho}_{U,z} \coloneqq \mathcal{M}_{\mathcal{U}}^{-1}\l( U^\dagger \op{z}{z} U \r)
\end{equation}
associated with the particular copy of $ \rho $ for which $ U $ was applied and $ \ket{z} $ was obtained. Classical shadows form an unbiased estimator for $ \rho $, and so they can be used to estimate the expectation value of any observable $ O $:
\begin{equation}
	\E_{U\sim\mathcal{U}, \ket{z} \sim U \rho U^\dagger} \l[ \tr(O \hat{\rho}_{U,z}) \r] = \tr(O \rho).
\end{equation}

The number of repetitions $ M $ required to obtain an accurate estimate for each $ \tr(O_j \rho) $ is controlled by the estimator's variance, which may be upper bounded by
\begin{equation}\label{eq:var_shadows}
    \max_{\text{states } \sigma} \mathbb{E}_{\substack{U \sim \mathcal{U} \\ \ket{z} \sim U \sigma U^\dagger}} \l[ \ev{z}{U \mathcal{M}_{\mathcal{U}}^{-1}(O_j) U^\dagger}{z}^2 \r] \eqqcolon \sns{O_j}{\mathcal{U}}^2.
\end{equation}
This quantity is referred to as the (squared) shadow norm. Then by median-of-means estimation, one may show that
\begin{equation}\label{eq:cs_M}
	M = \O\l( \frac{\log L}{\varepsilon^2} \max_{1 \leq j \leq L} \sns{O_j}{\mathcal{U}}^2 \r)
\end{equation}
samples suffice to estimate all expectation values to within additive error $ \varepsilon $. To minimize Eq.~\eqref{eq:cs_M} for a fixed collection of observables, the only available freedom is in $ \mathcal{U} $. One must therefore properly choose the ensemble of unitaries, with respect to the target observables.

\emph{Naive application to fermionic observables.}---A natural ensemble for near-term considerations is the group of single-qubit Clifford gates, $ \Cl(1)^{\otimes n} $ (i.e., Pauli measurements). For an $ \ell $-local Pauli observable $ P $, Huang \emph{et al.}~\cite{huang2020predicting} showed that $ \sns{P}{\Cl(1)^{\otimes n}}^2 = 3^\ell $, similar to the results of prior approaches also based on Pauli measurements~\cite{cotler2020quantum,bonet2019nearly,jiang2020optimal,evans2019scalable}. Although optimal for qubit $ \ell $-RDMs, such strategies cannot achieve the desired $ \O(n^k) $ scaling in the fermionic setting due to the inherent nonlocality of fermion-to-qubit mappings. Indeed, assuming that the $ n $ fermion modes are encoded into $ n $ qubits, the 1-degree Majorana operators necessarily possess an average qubit locality of at least $ \log_3(2n) $~\cite{jiang2020optimal}. This implies that, under random Pauli measurements, the squared shadow norm maximized over all $ 2k $-degree Majorana operators cannot do better than $ 3^{2k \log_3 (2n)} = 4^k n^{2k} $. In fact, for commonly used mappings such as the Jordan--Wigner~\cite{jordanwigner} or Bravyi--Kitaev~\cite{bravyi2002fermionic,seeley2012bravyi,tranter2015bravyi,havlivcek2017operator} transformations, the scalings are poorer ($ 3^n $ and $ {\sim} \, 9^k n^{3.2k} $, respectively).

\emph{Randomized measurements with fermionic Gaussian unitaries.}---To obtain optimal scaling in the shadow norm for fermionic observables, we propose randomizing over a different ensemble:~the group of fermionic Gaussian Clifford unitaries. First, the group of fermionic Gaussian unitaries $ \FGU(n) $ comprises all unitaries of the form
\begin{equation}
	U(e^{A}) \coloneqq \exp\l( -\frac{1}{4} \sum_{\mu,\nu=0}^{2n-1} A_{\mu\nu} \gamma_\mu \gamma_\nu \r),
\end{equation}
where $ A = -A^\T \in \R^{2n \times 2n} $. This condition implies that $ \FGU(n) $ is fully characterized by the Lie group $ \SO(2n) $~\cite{sattinger1986lie}. In particular, the adjoint action
\begin{equation}\label{eq:gaussian_adjoint_action}
	U(Q)^\dagger \gamma_\mu U(Q) = \sum_{\nu=0}^{2n-1} Q_{\mu\nu} \gamma_\nu \quad \forall Q \in \SO(2n)
\end{equation}
allows for efficient classical simulation of this group~\cite{knill2001fermionic,terhal2002classical,bravyi2004lagrangian,divincenzo2005fermionic,jozsa2008matchgates}. Second, the Clifford group $ \Cl(n) $ is the set of all unitary transformations that permute $ n $-qubit Pauli operators among themselves. It also admits an efficient classical representation~\cite{gottesman1998heisenberg,aaronson2004improved}. 

Because Majorana operators are equivalent to Pauli operators, we may deduce from Eq.~\eqref{eq:gaussian_adjoint_action} that a unitary that is both Gaussian and Clifford corresponds to $ Q $ being a signed permutation matrix. Note that this defines the full group of Majorana swap circuits~\cite{bonet2019nearly}. As the signs are irrelevant for our purpose, we simply consider the group of $ 2n \times 2n $ permutation matrices with determinant 1, known as (the faithful matrix representation of) the alternating group, $ \Alt(2n) $.

Concretely, we set
\begin{equation}\label{eq:FGU_ensemble}
	\mathcal{U}_{\FGU} \coloneqq \{ U(Q) \in \FGU(n) \mid Q \in \Alt(2n) \}.
\end{equation}
Given the context of fermionic tomography, the motivation for studying $ \FGU(n) $ is clear, as it preserves the degree of Majorana operators. On the other hand, the restriction to the discrete Clifford elements is valuable for practical considerations. As we show in Sec.~\ref{sec:fgu_appendix} of the Supplemental Material (SM), the permutational property of Clifford transformations necessarily implies that $ \mathcal{M}_{\FGU} $, as a linear map on the algebra of fermionic observables, is diagonal in the Majorana-operator basis,
\begin{equation}
    \mathcal{M}_{\FGU}(\Gamma_{\bm{\mu}}) = \lambda_{\bm{\mu}} \Gamma_{\bm{\mu}} \quad \forall \bm{\mu} \in \comb{2n}{2k},
\end{equation}
with eigenvalues
\begin{equation}
    \lambda_{\bm{\mu}} = \l. \binom{n}{k} \middle/ \binom{2n}{2k} \r. \equiv \lambda_{n,k}.
\end{equation}
In this diagonal form, the channel is readily invertible. Thus one may obtain closed-form expressions for the classical shadows $ \hat{\rho}_{Q,z} $, and, importantly, their corresponding estimators for $ \tr(\Gamma_{\bm{\mu}} \rho) $:
\begin{equation}\label{eq:ev_estimator}
    \tr(\Gamma_{\bm{\mu}} \hat{\rho}_{Q,z}) = \lambda_{n,k}^{-1} \sum_{\bm{\nu} \in \comb{2n}{2k}} \ev{z}{\Gamma_{\bm{\nu}}}{z} \det[ Q_{\bm{\nu},\bm{\mu}} ].
\end{equation}
Here, $ Q_{\bm{\nu},\bm{\mu}} $ denotes the submatrix of $ Q $ formed from its rows and columns indexed by $ \bm{\nu} $ and $ \bm{\mu} $, respectively~\cite{chapman2018classical}. Because $ Q $ is a permutation matrix, for each $ \bm{\mu} $ there is exactly one $ \bm{\nu}' $ such that $ \det[ Q_{\bm{\nu}',\bm{\mu}} ] \neq 0 $. Thus Eq.~\eqref{eq:ev_estimator} is nonzero if and only if that $ \Gamma_{\bm{\nu}'} $ is diagonal (i.e., maps to a Pauli-$ Z $ operator under a fermion-to-qubit transformation). In other words, the Clifford operation $ U(Q) $ sends $ \Gamma_{\bm{\mu}} $ to $ \pm\Gamma_{\bm{\nu}'} $, which can be estimated only if it is diagonal in the computational basis.

From Eq.~\eqref{eq:var_shadows}, the eigenvalues $ \lambda_{n,k}^{-1} $ of the inverse channel $ \mathcal{M}^{-1}_{\FGU} $ determine the shadow norm. The sample complexity of our approach then follows from Eq.~\eqref{eq:cs_M}. We summarize this first key result with the following theorem.

\begin{theorem}\label{thm:fgu_performance}
	Consider all $2k$-degree Majorana operators $\Gamma_{\bm{\mu}}$ on $n$ fermionic modes, labeled by $ \bm{\mu} \in \comb{2n}{2k} $. Under the ensemble $ \mathcal{U}_{\FGU} $ defined in Eq.~\eqref{eq:FGU_ensemble}, the shadow norm satisfies
	\begin{equation}\label{eq:norm_fgu}
		\sns{\Gamma_{\bm{\mu}}}{\FGU}^2 = \l. \binom{2n}{2k} \middle/ \binom{n}{k} \r. \approx \binom{n}{k}\sqrt{\pi k}
	\end{equation}
	for all $ \bm{\mu} \in \comb{2n}{2k} $. Thus the method of classical shadows estimates the fermionic $k$-RDM of any state $\rho$, i.e., $ \tr(\Gamma_{\bm{\mu}} \rho) \  \forall \bm{\mu} \in \bigcup_{j=1}^k \comb{2n}{2j} $, to additive error $ \varepsilon $, given
	\begin{equation}
		M = \O\l[ \binom{n}{k} \frac{k^{3/2} \log n}{\varepsilon^2} \r]
	\end{equation}
	copies of $ \rho $. Additionally, there is no subgroup $ G \subset \FGU(n) \cap \Cl(n) $ for which $ \sns{\Gamma_{\bm{\mu}}}{G} < \sns{\Gamma_{\bm{\mu}}}{\FGU} $.
\end{theorem}

The proof is presented in the SM, Sec.~\ref{sec:fgu_appendix}. Furthermore, noting from Eq.~\eqref{eq:ev_estimator} that $ \l|\tr(\Gamma_{\bm{\mu}} \hat{\rho}_{Q,z})\r| \leq \lambda_{n,k}^{-1} $, we also show in the SM that Bernstein's inequality~\cite{boucheron2013concentration} guarantees the above sample complexity via standard sample-mean estimation, rather than requiring the median-of-means technique proposed in the original work on classical shadows~\cite{huang2020predicting}.

This result has an intuitive conceptual interpretation. In the computational basis, there are precisely $ \binom{n}{k} $ diagonal Majorana operators within $ \comb{2n}{2k} $, corresponding to the unique $ k $-fold products of occupation-number operators (e.g., $ \prod_{j=1}^k a_{p_j}^\dagger a_{p_j} $) on $ n $ modes. As a permutation on $ \comb{2n}{2k} $, each element of $ \mathcal{U}_{\FGU} $ defines a different basis in which some other subset of $ \binom{n}{k} $ operators are diagonal. Then, one may expect to account for all $ |\comb{2n}{2k}| = \binom{2n}{2k} $ Majorana operators by randomly selecting on the order of $ \binom{2n}{2k}/\binom{n}{k} $ such bases;~Theorem~\ref{thm:fgu_performance} makes this claim rigorous.

Fermionic Gaussian circuits have a well-studied compilation scheme based on a Givens-rotation decomposition~\cite{wecker2015solving,kivlichan2018quantum,jiang2018quantum}. For a general element of $ \mathcal{U}_{\FGU} $, we require a circuit depth of at most $ 2n $ with respect to this decomposition~\cite{jiang2018quantum}. Additionally, as pointed out in Ref.~\cite{bonet2019nearly}, Gaussian unitaries commute with the global parity operator $ \Gamma_{(0,\ldots,2n-1)} $, allowing for error mitigation via symmetry verification~\cite{bonet2018low,mcardle2019error}.

Such compilation schemes make use of a group homomorphism property, $ U(Q_1) U(Q_2) = U(Q_1 Q_2) $. Therefore, if the circuit preparing $ \rho $ itself features fermionic Gaussian operations at the end, then we may further compile the measurement unitary into the state-preparation circuit~\cite{takeshita2020increasing}. In the case of indefinite particle number, this concatenation is essentially free. However, rotations with particle-number symmetry have depth at most $ n $~\cite{kivlichan2018quantum,jiang2018quantum}, so they must be embedded into the larger Gaussian unitary of depth $ 2n $. This observation motivates us to explore classical shadows over the number-conserving (NC) subgroup of $ \FGU(n) $.

\emph{Modification based on particle-number symmetry.}---Fermionic Gaussian unitaries that preserve particle number are naturally parametrized by $ \U(n) $. We express an element of this NC subgroup as
\begin{equation}
	U(e^\kappa) \coloneqq \exp\l( \sum_{p,q=0}^{n-1} \kappa_{pq} a_p^\dagger a_q \r),
\end{equation}
where $ \kappa = -\kappa^\dagger \in \C^{n \times n} $, hence $ e^{\kappa} \in \U(n) $. Because the particle-number symmetry manifests as a global phase factor $ e^{\tr\kappa/2} \in \U(1) $, without loss of generality we may consider $ \tr\kappa = 0 $, or equivalently, $ e^{\kappa} \in \SU(n) $. Such unitaries are also called orbital-basis rotations, owing to their adjoint action,
\begin{equation}
	U(u)^\dagger a_p U(u) = \sum_{q=0}^{n-1} u_{pq} a_{q} \quad \forall u \in \SU(n).
\end{equation}
This action on Majorana operators follows by linear extension.

Taking the intersection with the Clifford group requires that $ u $ be an $ n \times n $ generalized permutation matrix, with nonzero elements taking values in $ \{ \pm 1, \pm i \} $. This corresponds to the group of fermionic swap circuits~\cite{bravyi2002fermionic,kivlichan2018quantum}. Again, the phase factors on the matrix elements are irrelevant, so we shall restrict to $ u \in \Alt(n) $. By itself, this ensemble is insufficient to perform tomography. To see this, consider an arbitrary reduced density operator $ A_{\bm{p}}^\dagger A_{\bm{q}} \coloneqq a_{p_1}^\dagger \cdots a_{p_k}^\dagger a_{q_k} \cdots a_{q_1} $, where $ \bm{p}, \bm{q} \in \comb{n}{k} $. Such operators are diagonal in the computational basis only if $ \bm{p} = \bm{q} $. Informational completeness thus requires that there exists some $ U(u) $ that maps $ A_{\bm{p}}^\dagger A_{\bm{q}} $ to $ A_{\bm{r}}^\dagger A_{\bm{r}} $, for some $ \bm{r} \in \comb{n}{k} $. Because $ u \in \Alt(n) $, conjugation by $ U(u) $ simply permutes $ \bm{p} $ and $ \bm{q} $ independently. However, as permutations are bijective, it is not possible to permute both $ \bm{p} $ and $ \bm{q} $ to the same $ \bm{r} $ if $ \bm{p} \neq \bm{q} $.

Therefore, this ensemble will necessarily require operations beyond either the NC or Gaussian constraints. The simplest option for maintaining the low-depth structure of the basis rotations is to append Pauli measurements at the end of the circuit. Although the resulting circuit no longer preserves particle number, this addition incurs only a single layer of single-qubit gates. Specifically, we define the ensemble
\begin{equation}\label{eq:ncu_ensemble}
	\mathcal{U}_{\NC} \coloneqq \{ V \circ U(u) \mid V \in \Cl(1)^{\otimes n}, \,  u \in \Alt(n) \}.
\end{equation}
By virtue of introducing the notion of ``single-qubit'' gates, this method is dependent on the choice of fermion-to-qubit mapping. Let $ \loc(\Gamma_{\bm{\mu}}) $ denote the qubit locality of $ \Gamma_{\bm{\mu}} $ under some chosen mapping. While Pauli measurements incur a factor of $ 3^{\loc(\Gamma_{\bm{\mu}})} $ in the variance, the randomization over fermionic swap circuits effectively averages this quantity over all same-degree Majorana operators (rather than depending solely on the most nonlocal operator). Formally, we find that the shadow norm here is
\begin{equation}
	\sns{\Gamma_{\bm{\mu}}}{\NC}^2 = \E_{u \sim \Alt(n)} \l[ 3^{- \loc[ U(u)^\dagger \Gamma_{\bm{\mu}} U(u) ]} \r{]^{-1}}.
\end{equation}
Although this expression does not possess a closed form, the following theorem provides a universal upper bound.

\begin{theorem}\label{thm:NC_bounds}
	Under the ensemble $ \mathcal{U}_{\NC} $ defined in Eq.~\eqref{eq:ncu_ensemble}, the shadow norm obeys
	\begin{equation}
	\max_{\bm{\mu} \in \comb{2n}{2k}} \sns{\Gamma_{\bm{\mu}}}{\NC}^2 \leq \l. 9^k \binom{n}{2k} \middle/ \binom{n-k}{k} \r. = \O(n^k)
	\end{equation}
	for a fixed integer $ k $ and for all fermion-to-qubit mappings. Thus the method of classical shadows with $ \mathcal{U}_{\NC} $ estimates the $ k $-RDM to additive error $ \varepsilon $ with sample complexity
	\begin{equation}
	M = \O\l( \frac{n^k \log n}{\varepsilon^2} \r).
	\end{equation}
\end{theorem}

We provide derivations for the above results in the SM, Sec.~\ref{sec:ncu_appendix}.
Note that we have fixed $ k $ as a constant here, so the asymptotic notation may hide potentially large prefactors depending on $ k $. To understand such details, we turn to numerical studies.

\emph{Numerical calculations.}---Instead of drawing a new circuit for each repetition, here we employ a simplification more amenable to practical implementation. Fixing some integer $ r \geq 1 $, we generate a random collection $ \{ U^{(j)} \sim \mathcal{U} \}_{j=1}^{K_r} $ of $ K_r $ unitaries such that all target observables are covered at least $ r $ times. We say a Majorana operator $ \Gamma_{\bm{\mu}} $ is covered by the measurement unitary $ U $ if $ U \Gamma_{\bm{\mu}} U^\dagger $ is diagonal in the computational basis. Because the ensembles considered here consist of Gaussian and Clifford unitaries, we can determine all covered operators efficiently. Additionally, for the $ \mathcal{U}_{\NC} $ calculations, the qubit mappings were automated through OpenFermion~\cite{openfermion}.

To achieve precision corresponding to $ S = \O(1/\varepsilon^2) $ samples per observable, one repeats each circuit $ \lceil S/r \rceil $ times. The total number of circuit repetitions for our randomized protocols is then $ \lceil S/r \rceil K_r $. For practical purposes, we fix $ r = 50 $ in this work (see Sec.~\ref{sec:hyperparameter} of the SM for further details). To compare against prior deterministic strategies, we compute $S \times C$ for each such strategy, where $C$ is the number of sets of commuting observables constructed by a given strategy.

\begin{figure*}
    \includegraphics[width=\textwidth]{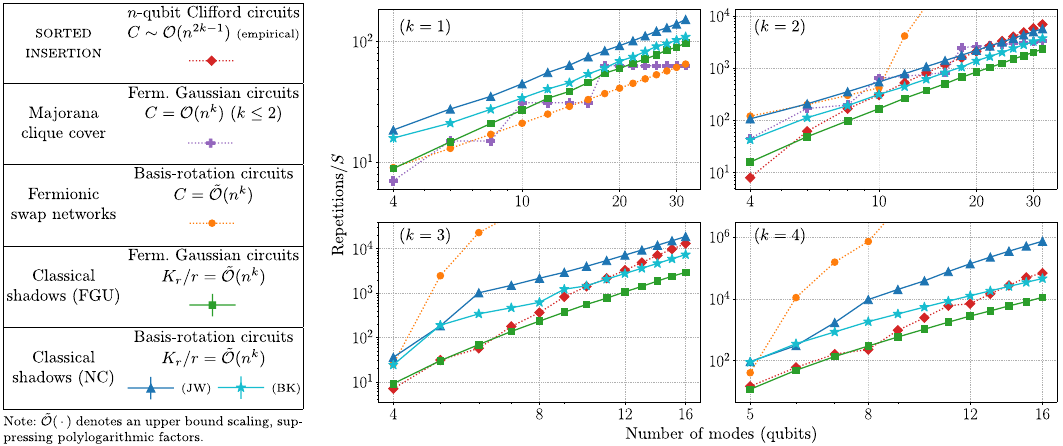}
    \caption{(Left)~Summary of the methods compared here, cataloging their required circuit types and scalings in the number of measurement settings. Because graph-based methods~\cite{jena2019pauli,yen2020measuring,gokhale2019minimizing,hamamura2020efficient} are resource intensive, we employ \textsc{sorted insertion}~\cite{crawford2019efficient} as a more tractable alternative. The Majorana clique cover~\cite{bonet2019nearly}, which employs the same class of fermionic Gaussian Clifford circuits as our classical shadows (FGU) unitaries, possesses optimal asymptotic scaling;~however, it exhibits jumps at powers of 2 due to a divide-and-conquer approach. Furthermore, the construction exists only for $ k \leq 2 $. The measurement strategy using fermionic swap networks is a generalization of the optimal 1-RDM strategy introduced in Ref.~\cite{arute2020hartree}, which we describe in Sec.~\ref{sec:rdmswap_appendix} of the SM. (Right)~Numerical performances  (log--log scale). Note that \textsc{sorted insertion} and the Majorana clique cover are equivalent for $ k = 1 $. Because our scheme uses randomization, we include error bars of 1 standard deviation, averaged over 10 instances. However, they are not visible at the scale of the plots, indicating the consistency of our method.}
    \label{fig:mainfig}
\end{figure*}

For the comparisons presented in Fig.~\ref{fig:mainfig}, we focus on the most competitive prior strategies applicable to fermionic RDM tomography. Because the 1-RDM has a relatively simple structure, optimal strategies are known~\cite{bonet2019nearly,arute2020hartree}, and so randomization underperforms for $ k = 1 $. However, the advantage of our $ \mathcal{U}_{\FGU} $-based method becomes clear for $ k \geq 2 $. When comparing against the Majorana clique cover, which features asymptotically optimal $ \O(n^2) $ scaling for the 2-RDM~\cite{bonet2019nearly}, we find a roughly twofold factor improvement by our approach.

For the $ \mathcal{U}_{\NC} $ case, we observe a trade-off between circuit size and measurement efficiency. As expected, the choice of fermion-to-qubit mapping matters here;~the Jordan--Wigner (JW) mapping performs worse than Bravyi--Kitaev (BK), as the former possesses more qubit nonlocality. Although more measurement settings are required compared to the $ \mathcal{U}_{\FGU} $ ensemble (e.g., a factor of $ {\sim} \, \text{2--5} $ under BK, depending on $ k $), each circuit itself requires only half the depth of general fermionic Gaussian circuits. Notably, however, $ \mathcal{U}_{\NC} $ classical shadows for the 2-RDM under the BK mapping is closely comparable to the Majorana clique cover.

\emph{Conclusions.}---We have adapted the framework of classical shadows to the efficient tomography of fermionic $k$-RDMs, applicable for all $k$. Numerical calculations demonstrate that our approach consistently outperforms prior strategies using measurement circuits of comparable sizes when $k \geq 2$, despite the logarithmic factor in the sample complexity (a consequence of rigorously bounding the worst-case probabilistic instances). The power of randomization here lies in avoiding the NP-hard problem of partitioning observables into commuting cliques~\cite{verteletskyi2020measurement,yen2020measuring,jena2019pauli,gokhale2019minimizing}. Instead, we show that a highly overlapping cover of the observables suffices to perform partial tomography efficiently, as a factor of $ \O(1/\varepsilon^2) $ repetitions is already required for this task.

An outlook for further applications is to adapt these ensembles, e.g., for Hamiltonian averaging. As expected, our method is less efficient in this context than those tailored for the task (see Sec.~\ref{sec:hamiltonian_appendix} of the SM for preliminary numerical calculations). Possible modifications may include biasing the distribution of unitaries~\cite{hadfield2020measurements,hillmich2021decision,hadfield2021adaptive,wu2021overlapped}, or derandomization techniques~\cite{huang2021efficient}.

We thank Hsin-Yuan (Robert) Huang and Charles Hadfield for helpful discussions. This work was supported partially by the National Science Foundation STAQ Project (No.~PHY-1818914),  QLCI Q-SEnSE Grant (No.~OMA-2016244), CHE-2037832, and the U.S.~Department of Energy, Office of Science, National Quantum Information Science Research Centers, Quantum Systems Accelerator. We acknowledge the use of high-performance computing resources provided by the UNM Center for Advanced Research Computing, supported in part by the National Science Foundation.

\onecolumngrid\newpage
\appendix

\tocless{\subsection*{\large{Supplemental Material}}}

(Revised:~October 3, 2022) An earlier version of this Supplemental Material contained an incorrect derivation of the variance for arbitrary observables (\cref{sec:arb_observable_appendix}). That section has been revised to provide an upper bound on the variance, which instead demonstrates its asymptotic scaling. The correct variance expression has been derived by Wan \emph{et al.}~\cite{wan2022matchgate} and O'Gorman~\cite{ogorman2022fermionic}, which we have used to amend our Hamiltonian averaging numerics (\cref{sec:hamiltonian_appendix}).

We thank these authors for bringing the error to our attention. The main conclusions of this work are unaffected by this revision.\\

\tableofcontents

\section{Additional notation}

First, we define some notation and preliminary concepts not discussed in the main text. For completeness, we generalize the definition of Majorana operators to include odd-degrees:
\begin{equation}
	\Gamma_{\bm{\mu}} \coloneqq (-i)^{\binom{k}{2}} \gamma_{\mu_1} \cdots \gamma_{\mu_k},
\end{equation}
where $ \bm{\mu} \in \comb{2n}{k} $. Where required, we may define the empty product ($ k = 0 $) as the identity, i.e., $ \Gamma_{\varnothing} \coloneqq \openone $. Then for each $ 0 \leq k \leq 2n $, we define the $ \R $-linear span of $ k $-degree Majorana operators on $ n $ fermion modes,
\begin{equation}
	\MA{n}{k} \coloneqq \spn\{\Gamma_{\bm{\mu}} \mid \bm{\mu} \in \comb{2n}{k} \},
\end{equation}
which is isomorphic (as a vector space) to $ \R^{\binom{2n}{k}} $. The algebra of all even-degree fermionic observables shall be denoted by
\begin{equation}
	\MA{n}{\mathrm{even}} \coloneqq \bigoplus_{k=0}^{n} \MA{n}{2k}.
\end{equation}
Due to the parity superselection rule~\cite{streater2000pct}, physical fermionic operators lie in $ \MA{n}{\mathrm{even}} $;~thus the notions of informational completeness for fermionic tomography are understood with respect to this algebra. To see this, suppose we interpret $ \ket{z} $ as a Fock basis state. Only products of occupation-number operators, $ \prod_p a_p^\dagger a_p $ (modulo anticommutation relations), have nonzero diagonal matrix elements in the Fock basis. Thus if $ U $ is a unitary generated by a fermionic Hamiltonian (hence respecting parity supersymmetry), only those $ O \in \MA{n}{\mathrm{even}} $ are able to satisfy $ \ev{z}{U O U^\dagger}{z} \neq 0 $. This argument also holds if one insists on viewing $ \ket{z} $ as a qubit computational basis state, since they are necessarily mapped to Fock basis states under a fermion-to-qubit encoding~\cite{steudtner2018fermion}.

The group of fermionic Gaussian unitaries $ \FGU(n) $ is generated by $ i\MA{n}{2} $, and its adjoint action on any $ k $-degree Majorana operator straightforwardly generalizes as
\begin{equation}\label{eq:gaussian_adjoint_action_k}
	U(Q)^\dagger \Gamma_{\bm{\mu}} U(Q) = \sum_{\bm{\nu}\in\comb{2n}{k}} \det[ Q_{\bm{\mu},\bm{\nu}} ] \Gamma_{\bm{\nu}},
\end{equation}
where $ Q_{\bm{\mu},\bm{\nu}} $ denotes the submatrix formed by taking the rows and columns of $ Q $ indexed by $ \bm{\mu} $ and $ \bm{\nu} $, respectively~\cite[Appendix A]{chapman2018classical}. 
This defines an adjoint representation $ \FGU(n) \to \SO[\mathfrak{su}(2^n)] \cong \SO(4^n - 1) $, understood in the sense that $ \FGU(n) \subseteq \SU(2^n) $. We will rather be interested in the orthogonal representation $ \Phi $ of $ \SO(2n) $ induced by this adjoint representation, through the group homomorphism $ U \colon \SO(2n) \to \FGU(n) $. That is, we define $ \Phi \colon \SO(2n) \to \SO(4^n - 1) $ by the matrix elements
\begin{equation}\label{eq:big_rep}
	[ \Phi(Q) ]_{\bm{\mu}\bm{\nu}} \coloneqq \det[ Q_{\bm{\mu},\bm{\nu}} ].
\end{equation}
By the Cauchy--Binet formula one may verify that this indeed defines an orthogonal matrix~\cite{chapman2018classical}. Additionally, $ \Phi $ inherits the homomorphism property from $ U $;~hence it is an orthogonal representation.

Since determinants are defined only for square matrices, $ \Phi $ possesses a natural decomposition as $ \Phi = \bigoplus_{k=1}^{2n} \phi_k $, where each $ \phi_k \colon \SO(2n) \to \SO(\MA{n}{k}) \cong \SO\l[ \binom{2n}{k} \r] $ is defined just as in \cref{eq:big_rep}, restricted a particular $ k $. These subrepresentations will be the main focus of our analysis.

Finally, because of their relation to the Clifford group, we make frequent use of permutations and their generalizations. We establish the relevant notation here. Let $ d,m > 0 $ be integers. We denote the symmetric group on $ d $ symbols by $ \Sym(d) $, which is faithfully represented by the group of $ d \times d $ permutation matrices. The alternating group $ \Alt(d) $ is the subgroup of all even parity permutations. The generalized symmetric group of cyclic order $ m $ over $ d $ symbols is defined via the wreath product, $ \Sym(m,d) \coloneqq \Z_m \wr \Sym(d) \equiv \Z_m^d \rtimes \Sym(d) $. Its faithful matrix representation is the group of $ d \times d $ generalized permutation matrices, wherein each nonzero matrix element can take on values from the $ m $th roots of unity. The determinant of such matrices is the sign of the underlying permutation, multiplied by all $ d $ nonzero elements. In particular, we shall refer to $ \Sym(2,d) $ as the group of signed permutation matrices, and denote its $ ({+1}) $-determinant subgroup by $ \Sym^{+}(2,d) $.

\section{\label{sec:fgu_appendix}Computations with the fermionic Gaussian Clifford ensemble}

We now derive the main results leading to \cref{thm:fgu_performance} of the main text. In \cref{sec:fgu_channel} we find an expression for the channel $ \mathcal{M}_{\FGU} $, showing that the permutational property of the Clifford group necessarily implies that the channel is diagonalized by the basis of Majorana operators [\cref{eq:fgu_sum_z}]. Then in \cref{sec:fgu_norm} we compute the corresponding eigenvalues of $ \mathcal{M}_{\FGU} $, which are directly related to the ensemble's shadow norm (\cref{lem:shadow_norm_bound}). In particular, \cref{lem:shadow_norm_bound} provides necessary and sufficient conditions for saturating the minimum value of the shadow norm under Gaussian Clifford measurements. In \cref{thm:tight_frame_bound} we explicitly calculate this minimum value, and finally with \cref{thm:G_irrep} we prove that both $ \FGU(n) \cap \Cl(n) $ and its $ \Alt(2n) $-generated subgroup $ \mathcal{U}_\FGU $ (the one presented in the main text) satisfy the necessary and sufficient conditions.

After proving these main results, in \cref{sec:class_shadow_appendix} we explicitly derive the expression for the classical shadows $ \hat{\rho}_{Q,z} $. For completeness, we show how to compute the shadow norm for an arbitrary observable in \cref{sec:arb_observable_appendix}. Finally, in \cref{sec:hoeffding} we show how the boundedness of classical shadow estimators and Bernstein's inequality allow one to avoid requiring median-of-means estimation, as mentioned in the main text.

\subsection{\label{sec:fgu_channel}The classical shadows channel}

Our goal is to find an analytic expression for the channel~\cite{huang2020predicting} $ \mathcal{M}_\FGU \colon \MA{n}{\mathrm{even}} \to \MA{n}{\mathrm{even}} $,
\begin{equation}
	\begin{split}
	\mathcal{M}_{\FGU}(O) &= \mathbb{E}_{U \sim \FGU(n) \cap \Cl(n)} \l[ \sum_{z\in\{0,1\}^n} \ev{z}{U O U^\dagger}{z} U^\dagger \op{z}{z} U \r]\\
	&= \mathbb{E}_{Q \sim \Sym^{+}(2,2n)} \l[ \sum_{z\in\{0,1\}^n} \ev{z}{U(Q) O U(Q)^\dagger}{z} U(Q)^\dagger \op{z}{z} U(Q) \r].
	\end{split}
\end{equation}
Since this is a linear map, we need only to evaluate it on a basis of $ \MA{n}{\mathrm{even}} $, the most natural choice being the Majorana operators. Distinguishing the Majorana operators which are diagonal (with respect to the computational basis) is highly important. For each $ 1 \leq k \leq n, $ we shall define the subset $ \diags{2n}{2k} \subseteq \comb{2n}{2k} $ of $ 2k $-combinations corresponding to the diagonal $ 2k $-degree Majorana operators. Formally,
\begin{equation}
	\diags{2n}{2k} \coloneqq \{ \bm{\mu} \in \comb{2n}{2k} \mid \ev{z}{\Gamma_{\bm{\mu}}}{z} \neq 0 \  \forall z \in \{0,1\}^n \}.
\end{equation}
Since there are $ \binom{n}{k} $ independent $ k $-fold products of occupation-number operators, each set has cardinality $ |\diags{2n}{2k}| = \binom{n}{k} $, so that
\begin{equation}
	\bigg| \bigcup_{k=1}^{n} \diags{2n}{2k} \bigg| = \sum_{k=1}^n \binom{n}{k} = 2^n - 1,
\end{equation}
which indeed matches the maximal number of simultaneously commuting Pauli operators~\cite{lawrence2002mutually} (e.g., all Pauli-$ Z $ operators of locality $ 1 $ to $ n $). For instance, under the Jordan--Wigner transformation, the corresponding sets are
\begin{equation}
\begin{split}
\diags{2n}{2} &\coloneqq \{ (p,p+1) \mid 0 \leq p \leq 2n-2 : p \text{ even} \},\\
\diags{2n}{4} &\coloneqq \{ (p,p+1,q,q+1) \mid 0 \leq p < q \leq 2n-2 : p,q \text{ even} \},\\
\diags{2n}{6} &\coloneqq \{ (p,p+1,q,q+1,r,r+1) \mid 0 \leq p < q < r \leq 2n-2 : p,q,r \text{ even} \},
\end{split}
\end{equation}
and so forth up to $ \diags{2n}{2n} $.

With this formalism, we can express basis states as
\begin{equation}
	\op{z}{z} = \frac{1}{2^n} \l( \openone + \sum_{j=1}^{n} \sum_{\bm{\mu} \in \diags{2n}{2j}} \ev{z}{\Gamma_{\bm{\mu}}}{z} \Gamma_{\bm{\mu}} \r),
\end{equation}
which undergo fermionic Gaussian evolution as
\begin{equation}\label{eq:UzzU}
	\begin{split}
	U(Q)^\dagger \op{z}{z} U(Q) &= \frac{1}{2^n} \l( \openone + \sum_{j=1}^{n} \sum_{\bm{\mu} \in \diags{2n}{2j}} \ev{z}{\Gamma_{\bm{\mu}}}{z} U(Q)^\dagger \Gamma_{\bm{\mu}} U(Q) \r)\\
	&= \frac{1}{2^n} \l( \openone + \sum_{j=1}^{n} \sum_{\bm{\mu} \in \diags{2n}{2j}} \ev{z}{\Gamma_{\bm{\mu}}}{z} \sum_{\bm{\nu} \in \comb{2n}{2j}} \det\l[ Q_{\bm{\mu},\bm{\nu}} \r] \Gamma_{\bm{\nu}} \r).
	\end{split}
\end{equation}
For the corresponding Born-rule probability factor, we use the fact that, since $ U $ is a group homomorphism, inverses are preserved [$ U(Q)^\dagger = U(Q^\T) $], and so for any $ \bm{\tau} \in \comb{2n}{2k} $ we have
\begin{equation}\label{eq:zUGUz}
	\begin{split}
	\ev{z}{U(Q) \Gamma_{\bm{\tau}} U(Q)^\dagger}{z} &= \sum_{\bm{\sigma} \in \comb{2n}{2k}} \det\l[ \l( Q^\T \r)_{\bm{\tau},\bm{\sigma}} \r] \ev{z}{\Gamma_{\bm{\sigma}}}{z}\\
	&= \sum_{\bm{\sigma} \in \diags{2n}{2k}} \det\l[ \l( Q^\T \r)_{\bm{\tau},\bm{\sigma}} \r] \ev{z}{\Gamma_{\bm{\sigma}}}{z}.
	\end{split}
\end{equation}
Now recall some basic properties of Majorana (equiv.~Pauli) operators. They are traceless,
\begin{equation}
	\sum_{z \in \{0,1\}^n} \ev{z}{\Gamma_{\bm{\sigma}}}{z} = \tr \Gamma_{\bm{\sigma}} = 0,
\end{equation}
and more generally they are Hilbert--Schmidt (trace) orthogonal,
\begin{equation}
	\tr\l( \Gamma_{\bm{\sigma}} \Gamma_{\bm{\mu}} \r) = 2^n \delta_{\bm{\sigma} \bm{\mu}}.
\end{equation}
In the case that $ \Gamma_{\bm{\sigma}} $ and $ \Gamma_{\bm{\mu}} $ are both diagonal, we have that
\begin{equation}
	\begin{split}
	\sum_{z \in \{0,1\}^n} \ev{z}{\Gamma_{\bm{\sigma}}}{z} \ev{z}{\Gamma_{\bm{\mu}}}{z} &= \sum_{z \in \{0,1\}^n} \ev{z}{\Gamma_{\bm{\sigma}} \Gamma_{\bm{\mu}}}{z}\\
	&= \tr\l( \Gamma_{\bm{\sigma}} \Gamma_{\bm{\mu}} \r) = 2^n \delta_{\bm{\sigma} \bm{\mu}}.
	\end{split}
\end{equation}
With these relations in hand, we multiply \cref{eq:UzzU,eq:zUGUz} and sum over all $ z\in\{0,1\}^n $ to obtain
\begin{equation}\label{eq:fgu_sum_z}
	\begin{split}
	\sum_{z\in\{0,1\}^n} \ev{z}{U(Q) \Gamma_{\bm{\tau}} U(Q)^\dagger}{z} U(Q)^\dagger \op{z}{z} U(Q) &= \frac{1}{2^n} \sum_{\bm{\sigma} \in \diags{2n}{2k}} \det\l[ \l( Q^\T \r)_{\bm{\tau},\bm{\sigma}} \r] \Bigg( \tr (\Gamma_{\bm{\sigma}}) \openone\\
	&\quad + \sum_{j=1}^{n} \sum_{\bm{\mu} \in \diags{2n}{2j}} 2^n \delta_{\bm{\sigma} \bm{\mu}} \sum_{\bm{\nu} \in \comb{2n}{2j}} \det\l[ Q_{\bm{\mu},\bm{\nu}} \r] \Gamma_{\bm{\nu}} \Bigg)\\
	&= \sum_{\substack{\bm{\sigma} \in \diags{2n}{2k} \\ \bm{\nu} \in \comb{2n}{2k}}} \det\l[ Q_{\bm{\sigma},\bm{\tau}} \r] \det\l[ Q_{\bm{\sigma},\bm{\nu}} \r] \Gamma_{\bm{\nu}}.
	\end{split}
\end{equation}
One may be tempted to use the Cauchy--Binet formula to evaluate the sum over $ \bm{\sigma} $;~crucially, however, the sum is restricted to $ \diags{2n}{2k} $, so the identity does not apply here. Instead, we will first evaluate the sum over $ \bm{\nu} $ by formalizing our notion of Gaussian Clifford transformations as degree-preserving permutations of Majorana operators. For generality, we state the following lemma with regards to any generalized permutation matrix.

\begin{lemma}\label{lem:unique_submatrix}
	Let $ Q \in \Sym(m,d) $ and fix $ \bm{\tau} \in \comb{d}{j} $, $ 1 \leq j \leq d $. Then there exists exactly one $ \bm{\sigma} \in \comb{d}{j} $ for which $ |{\det[ Q_{\bm{\sigma},\bm{\tau}} ]}| = 1 $;~otherwise, $ \det[ Q_{\bm{\sigma},\bm{\tau}} ] = 0 $. In particular, we have
	\begin{equation}\label{eq:subdet_lemma}
		\det[ ( Q^\dagger )_{\bm{\tau},\bm{\sigma}} ] \det[ Q_{\bm{\sigma},\bm{\nu}} ] = | {\det[ Q_{\bm{\sigma},\bm{\tau}} ]} | \, \delta_{\bm{\tau} \bm{\nu}}
	\end{equation}
	for all $ \bm{\nu} \in \comb{d}{j} $.
\end{lemma}

\begin{proof}
	By definition of permutation matrices, for each column $ q $ of $ Q $ there is exactly one row $ p $ for which $ Q_{pq} \neq 0 $. Furthermore, this row is unique to each column. This property generalizes from matrix elements to subdeterminants:~for each set of columns indexed by $ \bm{\tau} $, there is exactly one unique set of rows $ \bm{\sigma} $ for which $ Q_{\bm{\sigma},\bm{\tau}} $ has a nonzero element in each row. In other words, $ Q_{\bm{\sigma},\bm{\tau}} \in \Sym(m,j) $ and hence $ |{\det[ Q_{\bm{\sigma},\bm{\tau}} ]}| = 1 $ (recall that the determinant of these matrices is the sign of the underlying permutation multipled by $ m $th roots of unity). Otherwise, for all $ \bm{\sigma}' \neq \bm{\sigma} $, $ Q_{\bm{\sigma}',\bm{\tau}} $ possesses at least one row or column of all zeros, and hence has determinant $ 0 $.
	
	Next, we show that submatrices behave under conjugate transposition as $ ( Q_{\bm{\sigma},\bm{\tau}} )^\dagger = ( Q^\dagger )_{\bm{\tau},\bm{\sigma}} $, which can be seen by examining their matrix elements:
	\begin{equation}
	\begin{split}
		[ Q_{\bm{\sigma},\bm{\tau}} ]_{pq} &= Q_{\sigma_p \tau_q} = [ Q^\dagger ]_{\tau_q \sigma_p}^* = [ ( Q^\dagger )_{\bm{\tau},\bm{\sigma}} ]_{qp}^* = [ ( [ Q^\dagger ]_{\bm{\tau},\bm{\sigma}} )^\dagger ]_{pq}.
	\end{split}
	\end{equation}
	Since the determinant is invariant under transposition and preserves complex conjugation, $ \det[ ( Q^\dagger )_{\bm{\tau},\bm{\sigma}} ] = \det[ ( Q_{\bm{\sigma},\bm{\tau}} ) ]^* $. This gives us
	\begin{equation}
		\det[ ( Q^\dagger )_{\bm{\tau},\bm{\sigma}} ] \det[ Q_{\bm{\sigma},\bm{\nu}} ] = \det[ ( Q_{\bm{\sigma},\bm{\tau}} ) ]^* \det[ Q_{\bm{\sigma},\bm{\nu}} ].
	\end{equation}
	But since $ \bm{\tau} $ is the unique $ j $-combination for which $ \det[ ( Q_{\bm{\sigma},\bm{\tau}} ) ] \neq 0 $ for a fixed $ \bm{\sigma} $, the above expression can only be nonzero when $ \bm{\nu} = \bm{\tau} $. Equation~\eqref{eq:subdet_lemma} thus follows.
\end{proof}

\begin{remark}\label{rem:unique_submatrix}
    This uniqueness property of nonzero subdeterminants implies that the image of the orthogonal representation $ \Phi \colon \SO(2n) \to \SO(4^n - 1) $, when restricted to a subgroup $ G \subseteq \Sym^{+}(2,2n) \subset \SO(2n) $, lies in $ \Sym^{+}(2,4^n - 1) $. Its subrepresentations also satisfy $ \phi_k(G) \subseteq \Sym\l[2,\binom{2n}{k}\r] $. From a physical perspective, these are all straightforward consequences of the Clifford property imposed on our unitary ensemble.
\end{remark}

Applying \cref{lem:unique_submatrix} to \cref{eq:fgu_sum_z} reveals that the Majorana operators are in fact the eigenbasis of $ \mathcal{M}_\FGU $:
\begin{equation}\label{eq:M_eigenvalues_eq}
	\begin{split}
	\mathcal{M}_\FGU(\Gamma_{\bm{\tau}}) &= \E_{Q \sim \Sym^{+}(2,2n)} \l[ \sum_{z\in\{0,1\}^n} \ev{z}{U(Q) \Gamma_{\bm{\tau}} U(Q)^\dagger}{z} U(Q)^\dagger \op{z}{z} U(Q) \r]\\
	&= \E_{Q \sim \Sym^{+}(2,2n)} \l[ \sum_{\bm{\sigma} \in \diags{2n}{2k}} | {\det[ Q_{\bm{\sigma},\bm{\tau}} ]} | \r] \Gamma_{\bm{\tau}}.
	\end{split}
\end{equation}

\subsection{\label{sec:fgu_norm}The shadow norm}

To evaluate the eigenvalues of \cref{eq:M_eigenvalues_eq}, we invoke the theory of finite frames~\cite{han2000frames,waldron2018introduction}. We begin with the definition of a frame.

\begin{definition}
	Let $ V $ be a Hilbert space with inner product $ \langle \cdot, \cdot \rangle $. A \emph{frame} is a sequence $ \{ x_j \}_j \subset V $ which satisfies
	\begin{equation}
	\alpha \| v \|_V^2 \leq \sum_{j} | \langle x_j, v \rangle |^2 \leq \beta \| v \|_V^2 \quad \forall v \in V,
	\end{equation}
	for some real constants $ \alpha,\beta > 0 $ (called the \emph{frame bounds}). Here, $ \| v \|_V \coloneqq \sqrt{\langle v, v \rangle} $. A frame is called \emph{tight} if $ \alpha=\beta $.
\end{definition}

Recall that $ \det[ Q_{\bm{\sigma},\bm{\tau}} ] $ defines the matrix elements of the orthogonal representation $ \phi_{2k} \colon \SO(2n) \to \SO\l[ \binom{2n}{2k} \r] $. In order to demonstrate the optimality of our ensemble, we shall generalize \cref{eq:M_eigenvalues_eq} to take the average over any subgroup $ G \subseteq \Sym^{+}(2,2n) $, in which case we consider the restricted representations $ \phi_{2k}\big|_G \colon G \to \SO\l[ \binom{2n}{2k} \r] $. (Whenever the context is clear, we shall simply write $ \phi_{2k} $.) To simplify notation, we take $ \R^{\binom{2n}{2k}} $ as our representation space, spanned by the standard basis $ \{ e_{\bm{\mu}} \mid \bm{\mu} \in \comb{2n}{2k} \} $.

Consider the group orbit
\begin{equation}
	\phi_{2k}(G) e_{\bm{\sigma}} \coloneqq \{ \phi_{2k}(Q) e_{\bm{\sigma}} \mid Q \in G \}.
\end{equation}
A sufficient condition for the eigenvalues of $ \mathcal{M}_G $ to be nonzero (hence guaranteeing the existence of $ \mathcal{M}_G^{-1} $) is that $ \phi_{2k}(G) e_{\bm{\sigma}} $ be a frame for some $ \bm{\sigma} \in \diags{2n}{2k} $. Indeed, assuming the frame condition, we have
\begin{equation}
	\begin{split}
	\E_{Q \sim G} | {\det[ Q_{\bm{\sigma},\bm{\tau}} ]} | &= \frac{1}{|G|} \sum_{Q \in G} | {\det[ Q_{\bm{\sigma},\bm{\tau}} ]} |^2\\
	&= \frac{1}{|G|} \sum_{Q \in G} | \langle \phi_{2k}(Q)^\T e_{\bm{\sigma}} , e_{\bm{\tau}} \rangle |_2^2\\
	&\geq \frac{\alpha}{|G|} \| e_{\bm{\tau}} \|_2^2 > 0
	\end{split}
\end{equation}
for all $ \bm{\tau} \in \comb{2n}{2k} $. Note that the transpose is irrelevant since $ G $ is a group and $ \phi_{2k} $ an orthogonal representation. While $ \mathcal{M}_G $ may be made positive even without taking $ G $ to be a group, the group structure has desirable implications for the measurement complexity (shadow norm) of our scheme. (It is also generally easier to sample from a well-established group like the symmetric group, rather than some ad hoc subset of $ \Sym^{+}(2,2n) $.) To see this, we shall focus attention exclusively to tight frames, which is motivated by the following observation.

\begin{lemma}\label{lem:shadow_norm_bound}
	Let $ G \subseteq \Sym^{+}(2,2n) $ be a subset such that $ \phi_{2k}(G) e_{\bm{\sigma}} $ is a frame for all $ \bm{\sigma} \in \diags{2n}{2k} $, $ 1 \leq k \leq n $. Let $ \alpha_k,\beta_k > 0 $ be the cumulative frame bounds over all $ \bm{\sigma} $ for each $ k $, i.e.,
	\begin{equation}\label{eq:general_frame_for_eigenvalues}
		\alpha_k \leq \frac{1}{|G|} \sum_{\bm{\sigma} \in \diags{2n}{2k}} \sum_{Q \in G} | {\det[ Q_{\bm{\sigma},\bm{\tau}} ]} |^2 \leq \beta_k
	\end{equation}
	for all $ \bm{\tau} \in \comb{2n}{2k} $. Then the shadow norm associated with the ensemble $ G $ satisfies
	\begin{equation}
		\alpha_k / \beta_k^2 \leq \sns{\Gamma_{\bm{\tau}}}{G}^2 \leq \beta_k/\alpha_k^{2}.
	\end{equation}
	These bounds are saturated if and only if the cumulative frame is tight, i.e., $ \sns{\Gamma_{\bm{\tau}}}{G}^2 = \alpha_k^{-1} $.
\end{lemma}

\begin{remark}
	In \cref{eq:general_frame_for_eigenvalues}, and hence in \cref{eq:M_eigenvalues_eq}, the actual frame being used is $ \phi_{2k}(G) (e_{\bm{\sigma}} |G|^{-1/2}) $, i.e., the ``orbit'' of the basis vector scaled by a factor of $ |G|^{-1/2} $. (We use scare quotes to indicate that, since we do not require $ G $ to be a group here, the corresponding set of vectors may not be an orbit proper.)
\end{remark}

\begin{proof}
	First, we recall the definition of the shadow norm:
	\begin{equation}
			\sns{\Gamma_{\bm{\tau}}}{G}^2 = \max_{\text{states } \rho} \l( \mathbb{E}_{Q \sim G} \l[ \sum_{z \in \{0,1\}^n} \ev{z}{U(Q) \rho U(Q)^\dagger}{z} \ev{z}{U(Q) \mathcal{M}^{-1}(\Gamma_{\bm{\tau}}) U(Q)^\dagger}{z}^2 \r] \r).
	\end{equation}
	Adapting \cref{eq:M_eigenvalues_eq} to the language of frames yields the inequalities
	\begin{equation}
		\beta_k^{-1} \leq \| \mathcal{M}_G^{-1}(\Gamma_{\bm{\tau}}) \| \leq \alpha_k^{-1},
	\end{equation}
	where we have used the fact that the spectral norm $ \| \cdot \| $ of the Majorana operators is $ 1 $. Thus
	\begin{equation}
		\begin{split}
		\sns{\Gamma_{\bm{\tau}}}{G}^2 &\leq \alpha_k^{-2} \max_{\text{states } \rho} \l( \mathbb{E}_{Q \sim G} \l[ \sum_{z \in \{0,1\}^n} \ev{z}{U(Q) \rho U(Q)^\dagger}{z} \ev{z}{U(Q) \Gamma_{\bm{\tau}} U(Q)^\dagger}{z}^2 \r] \r)\\
		&= \alpha_k^{-2} \max_{\text{states } \rho} \l( \mathbb{E}_{Q \sim G} \l[ \sum_{z \in \{0,1\}^n} \ev{z}{U(Q) \rho U(Q)^\dagger}{z} \l( \sum_{\bm{\mu} \in \comb{2n}{2k}} \det[ ( Q^\T )_{\bm{\tau},\bm{\mu}} ] \ev{z}{\Gamma_{\bm{\mu}}}{z} \r{)^2} \r] \r).
		\end{split}
	\end{equation}
	Let us examine the innermost bracketed term. Since $ \ev{z}{\Gamma_{\bm{\mu}}}{z} = \pm 1 $ when $ \bm{\mu} \in \diags{2n}{2k} $ and $ 0 $ otherwise, we can restrict the sum to run over $ \diags{2n}{2k} $. Furthermore, by \cref{lem:unique_submatrix}, we know that $ \det[ ( Q^\T )_{\bm{\tau},\bm{\mu}} ] $ is nonzero only for one particular $ \bm{\mu} $, which may or may not lie in $ \diags{2n}{2k} $. This means that the sum contains at most one nonzero term, which takes the value $ \pm 1 $, allowing us to ``transform'' the square into an absolute value:
	\begin{equation}\label{eq:norm_var_term}
		\l( \sum_{\bm{\mu} \in \comb{2n}{2k}} \det[ ( Q^\T )_{\bm{\tau},\bm{\mu}} ] \ev{z}{\Gamma_{\bm{\mu}}}{z} \r{)^2} = \sum_{\bm{\mu} \in \diags{2n}{2k}} | {\det[ Q_{\bm{\mu},\bm{\tau}} ]} |.
	\end{equation}
	Importantly, this quantity does not depend on $ z $, which makes the sum over $ z $ trivial (as well as the maximum over all states):
	\begin{equation}
		\sum_{z \in \{0,1\}^n} \ev{z}{U(Q) \rho U(Q)^\dagger}{z} = \tr \rho = 1.
	\end{equation}
	We are then left with taking the average over $ Q \sim G $ in \cref{eq:norm_var_term}, which is simply the original quantity of interest obeying the frame bounds [\cref{eq:general_frame_for_eigenvalues}]. Therefore
	\begin{equation}
		\begin{split}
			\sns{\Gamma_{\bm{\tau}}}{G}^2 &\leq \alpha_k^{-2} \E_{Q \sim G} \l[ \sum_{\bm{\mu} \in \diags{2n}{2k}} | {\det[ Q_{\bm{\mu},\bm{\tau}} ]} | \r]\\
			&\leq \alpha_k^{-2} \beta_k.
		\end{split}
	\end{equation}
	To obtain the lower bound, simply reverse the roles of $ \alpha_k $ and $ \beta_k $. Since the only inequalities invoked were those of the frame bounds, it follows that equality holds in both directions if and only if the frame is tight.
\end{proof}

\cref{lem:shadow_norm_bound} is useful in two ways:~first, it gives an estimate on the shadow norm for any valid subset of fermionic Gaussian Clifford unitaries, assuming one has bounds on the eigenvalues of $ \mathcal{M}_G $. Second, it tells us that those subsets which give rise to \emph{tight} frames exhibit optimal sample complexity in the sense of the shadow norm. Tight frames generated by the action of finite groups have been completely characterized through representation theory~\cite{vale2004tight,waldron2018introduction}. Below we restate the primary results relevant to our context, which will motivate us to restrict $ G $ to a group.

\begin{proposition}[{\cite[Theorems 6.3 and 6.5]{vale2004tight}}]
	\label{prop:tight_frames}
	Let $ H $ be a finite group, $ V $ a Hilbert space, and $ \varphi \colon H \to \U(V) $ a unitary representation. Then:
	\begin{enumerate}
		\item Every orbit $ \varphi(H)v $, $ v \in V \setminus \{0\} $, is a tight frame if and only if every orbit spans $ V $ [i.e., $ \varphi $ is an irreducible representation (irrep)].
		\item There exists $ v \in V \setminus \{0\} $ for which $ \varphi(H)v $ is a tight frame if and only if there exists $ w \in V \setminus \{0\} $ such that $ \spn(\varphi(H)w) = V $.
	\end{enumerate}
\end{proposition}

In our context, this means that if $ \phi_{2k}\big|_G $ is irreducible, then every orbit is tight, hence saturating the bound of \cref{lem:shadow_norm_bound}. Alternatively, if $ \phi_{2k}\big|_G $ is not an irrep, there may be only specific orbits which form tight frames. Fortunately, this complication does not arise in our setting due to the fact that we are dealing exclusively with (signed) permutation matrices. We formalize this notion with the following lemma. (Although we are only interested in the even-degree representations, we formulate the statement to apply to all $ 1 \leq k \leq 2n $ for completeness.)

\begin{lemma}\label{lem:all_orbits_span}
	Let $ G\subseteq\Sym^{+}(2,2n) $ be a group such that $ \spn(\phi_{k}(G) v_0) = \R^{\binom{2n}{k}} $
	for some nonzero $ v_0\in\R^{\binom{2n}{k}} $. Then $ \spn(\phi_{k}(G) v) = \R^{\binom{2n}{k}} $ for all nonzero $ v\in\R^{\binom{2n}{k}} $.
\end{lemma}
\begin{proof}
	Without loss of generality, we can consider $ v_0 $ an arbitrary element of the standard basis. From \cref{rem:unique_submatrix}, we have $ \phi_k(G) \subseteq \Sym\l[2,\binom{2n}{k}\r] $;~therefore, $ \spn(\phi_k(G)v_0) = \R^{\binom{2n}{k}} $ if and only if the orbit is the entire basis (modulo signs):
	\begin{equation}
	\phi_k(G) v_0 = \{ e_{\bm{\mu}} \text{ and/or } {-}e_{\bm{\mu}} \mid \bm{\mu} \in \comb{2n}{k} \}.
	\end{equation}
	Now consider some other basis vector $ w_0 \neq v_0 $. Since there exists some $ g\in G $ such that $ \phi_k(g)v_0 = \pm w_0 $, it follows that
	\begin{equation}
	\begin{split}
		\phi_k(G)v_0 &= \phi_k(G) [\phi_k(g^{-1}) (\pm w_0)]\\
		&= \phi_k(G g^{-1}) (\pm w_0) = \phi_k(G) (\pm w_0).
	\end{split}
	\end{equation}
	Since the sign is irrelevant when taking the span, we see that the orbit of any basis vector spans the space. Hence the orbit of every nonzero vector does as well.
\end{proof}

This result allows us to ignore the second part of \cref{prop:tight_frames}, so that we only have to consider the irreducibility of $ \phi_{2k}\big|_G $. Furthermore, we do not have to worry about choosing some particular $ \bm{\sigma} \in \diags{2n}{2k} $ to generate our cumulative frame;~although the specific form of these tuples changes depending on the choice of fermion-to-qubit mapping, the above results tell us that the behavior of such tight frames is uniform across all of $ \comb{2n}{2k} $.

Enumerating all possible subgroups to search for the one which gives the largest frame bound (hence smallest shadow norm) is a highly impractical task. Fortunately, this is not necessary, as it turns out that \emph{all} irreps $ \phi_{2k}\big|_G $ yield the same frame bound. To see this, we introduce an alternative, equivalent formulation of frames based on the \emph{frame operator} $ T \colon V \to V $,
\begin{equation}
	T \coloneqq \sum_{j} \langle x_j, \cdot \, \rangle x_j,
\end{equation}
where $ \{x_j\}_j $ is a frame for $ V $.

\begin{proposition}[{\cite[Theorem 3]{cotfas2010finite}}]\label{prop:frame_op}
	Let $ H $ be a finite group, $ V $ a Hilbert space, and $ \varphi \colon H \to \U(V) $ an irreducible unitary representation. For any nonzero $ v \in V $, the frame operator for the orbit $ \varphi(H)v $ is
	\begin{equation}
		T = \frac{|\varphi(H)v|}{\dim V} \| v \|_V^2 \, \emph{\openone}_{V},
	\end{equation}
	where $ \emph{\openone}_{V} $ is the identity operator on $ V $.
\end{proposition}

In general, $ T = \alpha\openone_V $ if and only if the frame is tight (with frame bound $ \alpha $). In the context of tight frames generated by irreps, this result is essentially a variant on Schur's lemma. The utility of \cref{prop:frame_op} in particular is that the frame bound is given explicitly. Our claim that all irreducible subgroups yield the same shadow norm then follows.

\begin{theorem}\label{thm:tight_frame_bound}
	Let $ G \subseteq \Sym^{+}(2,2n) $ be irreducible with respect to $ \phi_{2k} $. Then
	\begin{equation}\label{eq:frame_eigenvalue}
	\begin{split}
		\E_{Q \sim G} \l[ \sum_{\bm{\sigma} \in \diags{2n}{2k}} | {\det[ Q_{\bm{\sigma},\bm{\tau}} ]} | \r] &\equiv \frac{1}{|G|} \sum_{Q \in G} \sum_{\bm{\sigma} \in \diags{2n}{2k}} | [ \phi_{2k}(Q) ]_{\bm{\sigma} \bm{\tau}} |^2\\
		&= \l. \binom{n}{k} \middle/ \binom{2n}{2k} \r.
	\end{split}
	\end{equation}
	for all $ \bm{\tau} \in \comb{2n}{2k} $.
\end{theorem}

\begin{proof}
	Consider the orbit $ \phi_{2k}(G) (e_{\bm{\sigma}} |G|^{-1/2}) $. Let $ \Stab_{G}(e_{\bm{\sigma}}) \coloneqq \{ Q \in G \mid \phi_{2k}(Q) e_{\bm{\sigma}} = e_{\bm{\sigma}} \} $ be its stabilizer subgroup. By the orbit--stabilizer theorem and Lagrange's theorem~\cite{grillet2007abstract},
	\begin{equation}
		| \phi_{2k}(G) (e_{\bm{\sigma}} |G|^{-1/2}) | = \frac{|G|}{| \Stab_{G}(e_{\bm{\sigma}}) |}.
	\end{equation}
	Using \cref{prop:frame_op} and recalling that the Hilbert space $ V $ of our frame is $ \R^{\binom{2n}{2k}} $, we see that the corresponding frame bound is
	\begin{equation}
	\begin{split}
		\alpha_{\bm{\sigma}} = \| T_{\bm{\sigma}} \| &= \frac{|G|}{| \Stab_{G}(e_{\bm{\sigma}}) |} \binom{2n}{2k}^{-1} \l\| \frac{e_{\bm{\sigma}}}{|G|^{1/2}} \r{\|_2^2}\\
		&= \frac{1}{| \Stab_{G}(e_{\bm{\sigma}}) |} \binom{2n}{2k}^{-1}.
	\end{split}
	\end{equation}
	One must be careful when handling this $ | \Stab_{G}(e_{\bm{\sigma}}) | $ term, since, when constructing the frame operator, we sum over all elements of the orbit, rather than of the group. However, the particular sum that we are interested in, namely \cref{eq:frame_eigenvalue}, \emph{is} over the entire group, and thus precisely double counts the elements of the stabilizer subgroup:
	\begin{equation}
	\begin{split}
	\frac{1}{|G|} \sum_{Q \in G} | [ \phi_{2k}(Q) ]_{\bm{\sigma} \bm{\tau}} |^2 &= \sum_{Q \in G} | \langle \phi_{2k}(Q)^\T e_{\bm{\sigma}} |G|^{-1/2} , e_{\bm{\tau}} \rangle |^2\\
	&= | \Stab_{G}(e_{\bm{\sigma}}) | \sum_{w \in \phi_{2k}(G) (e_{\bm{\sigma}} |G|^{-1/2})} | \langle w, e_{\bm{\tau}} \rangle |^2\\
	&= | \Stab_{G}(e_{\bm{\sigma}}) | \alpha_{\bm{\sigma}} = \binom{2n}{2k}^{-1}.
	\end{split}
	\end{equation}
	Since this quantity does not depend on $ \bm{\sigma} $, the remaining sum over $ \diags{2n}{2k} $ simply incurs a factor of $ |\diags{2n}{2k}| = \binom{n}{k} $.
\end{proof}

\begin{corollary}
	In conjunction with \cref{lem:shadow_norm_bound}, it immediately follows that
	\begin{equation}
		\sns{\Gamma_{\bm{\tau}}}{G}^2 = \l. \binom{2n}{2k} \middle/ \binom{n}{k} \r.
	\end{equation}
	for all irreducible $ G \subseteq \Sym^{+}(2,2n) $.
\end{corollary}

Finally, in order to obtain a concrete example of such tight frames, we show that $ \Sym^{+}(2,2n) $ is irreducible with respect to $ \phi_{k} $ for all $ 1 \leq k \leq 2n $. Again, though we are only interested in the even-degree case, we prove the statement for general $ k $ for completeness. It is then straightforward to show that $ \Alt(2n) \subset \Sym^{+}(2,2n) $ is also irreducible.

\begin{theorem}\label{thm:G_irrep}
	Let $ G = \Sym^{+}(2,2n) $ or $ \Alt(2n) $. For all $ 1 \leq k \leq 2n $, $ \phi_k\big|_G $ is irreducible.
\end{theorem}

\begin{proof}
	We begin with the case $ G = \Sym^{+}(2,2n) $;~it will be apparent that the proof methods adapt fully to the $ G = \Alt(2n) $ case. We show irreducibility via a standard result of character theory~\cite{fulton2004representation}:~$ \phi_k\big|_G $ is an irrep if and only if
	\begin{equation}
		\frac{1}{|G|} \sum_{Q \in G} \tr[ \phi_k(Q) ]^2 = 1.
	\end{equation}
	Here, the trace of our representation is simply
	\begin{equation}
		\tr[ \phi_k(Q) ] = \sum_{\bm{\mu} \in \comb{2n}{k}} \det[ Q_{\bm{\mu},\bm{\mu}} ].
	\end{equation}
	Expanding the square yields
	\begin{equation}\label{eq:expanded_char}
		\frac{1}{|G|} \sum_{Q \in G} \tr[ \phi_k(Q) ]^2 = \frac{1}{|G|} \sum_{Q \in G} \l( \sum_{\bm{\mu} \in \comb{2n}{k}} \det[ Q_{\bm{\mu},\bm{\mu}} ]^2 + \sum_{\substack{\bm{\mu},\bm{\nu} \in \comb{2n}{k} \\ \bm{\mu} \neq \bm{\nu}}} \det[ Q_{\bm{\mu},\bm{\mu}} ] \det[ Q_{\bm{\nu},\bm{\nu}} ] \r).
	\end{equation}
	We will calculate the average of the each term separately.
	
	Since $ | { \det[ Q_{\bm{\mu},\bm{\mu}} ] } | \in \{0,1\} $, the diagonal sum is
	\begin{equation}
		\frac{1}{|G|} \sum_{Q \in G} \sum_{\bm{\mu} \in \comb{2n}{k}} \det[ Q_{\bm{\mu},\bm{\mu}} ]^2 = \sum_{\bm{\mu} \in \comb{2n}{k}} \E_{Q \sim G} | { \det[ Q_{\bm{\mu},\bm{\mu}} ] } |.
	\end{equation}
	Intuitively, $ \E_{Q \sim G} | { \det[ Q_{\bm{\mu},\bm{\mu}} ] } | $ is the ``density'' of the signed permutation matrices whose $ (\bm{\mu}, \bm{\mu}) $ submatrix has nonzero subdeterminant. To compute this average, we proceed by a visual argument using the structure of the matrix. Recall that $ \det[ Q_{\bm{\mu},\bm{\mu}} ] \neq 0 $ if and only if $ Q_{\bm{\mu},\bm{\mu}} \in \Sym(2,k) $. When writing down a matrix as an array, we can choose any ordering of the row and column indices, so long as this choice is consistent. Therefore, we shall order the indices such that $ \bm{\mu} $ lies in the first $ k $ rows/columns:
	\begin{equation}
		Q = \l( \begin{array}{c | c}
		Q_{\bm{\mu},\bm{\mu}} & * \\
		\hline
		* & *
		\end{array} \r).
	\end{equation}
	Requiring that $ Q_{\bm{\mu},\bm{\mu}} \in \Sym(2,k) $ immediately sets the off-diagonal blocks to be all zeros, hence such a $ Q $ is block diagonal in this ordering of rows and columns. Furthermore, since $ \det Q = 1 $, the remaining $ (2n-k) \times (2n-k) $ block must have the same determinant as the $ Q_{\bm{\mu},\bm{\mu}} $ block. In other words, a $ Q $ which satisfies $ \det[ Q_{\bm{\mu},\bm{\mu}} ] \neq 0 $ must take the form
	\begin{equation}
		Q \in \l( \begin{array}{c | c}
		\Sym^{+}(2,k) & 0 \\
		\hline
		0 & \Sym^{+}(2,2n-k)
		\end{array} \r)
		\cup
		\l( \begin{array}{c | c}
		\Sym^{-}(2,k) & 0 \\
		\hline
		0 & \Sym^{-}(2,2n-k)
		\end{array} \r),
	\end{equation}
	where $ \Sym^{-}(2,d) \coloneqq \{ R \in \Sym(2,d) \mid \det R = -1 \} $. Although $ \Sym^{-}(2,d) $ is not a group, it has the same number of elements as $ \Sym^{+}(2,d) $:~$ |\Sym^{-}(2,d)| = |\Sym^{+}(2,d)| = 2^d d!/2 $. The density is thus
	\begin{equation}\label{eq:avg_density_calc}
	\begin{split}
		\E_{Q \sim G} | { \det[ Q_{\bm{\mu},\bm{\mu}} ] } | &= \frac{ |\Sym^{+}(2,k) \oplus \Sym^{+}(2,2n-k)| + |\Sym^{-}(2,k) \oplus \Sym^{-}(2,2n-k)| }{ |\Sym^{+}(2,2n)| }\\
		&= \frac{ 2 \l( 2^k k! \, 2^{2n-k} (2n-k)! \r) / 4 }{ 2^{2n} (2n)! / 2 } = \binom{2n}{k}^{-1},
	\end{split}
	\end{equation}
	and so
	\begin{equation}
		\sum_{\bm{\mu} \in \comb{2n}{k}} \E_{Q \sim G} | { \det[ Q_{\bm{\mu},\bm{\mu}} ] } | = 1.
	\end{equation}
	
	Next we show that the off-diagonal sum of \cref{eq:expanded_char} vanishes. The argument follows by generalizing the above calculation. For $ \bm{\mu} \neq \bm{\nu} $, the tuples (thought of as sets) may overlap $ 0 \leq j \leq k-1 $ times. We order the matrix representation such that $ \bm{\mu} $ makes up the first $ k $ rows/columns as before, but additionally the overlapping indices $ \bm{\mu} \cap \bm{\nu} $ are placed at the last $ j $ spots of this $ k \times k $ block. We then place the remaining part of $ \bm{\nu} $ in the following $ (k-j) \times (k-j) $ block. Visually, we have
	\begin{equation}
		Q = \l( \begin{array}{c | c | c | c}
		Q_{\bm{\mu}\setminus\bm{\nu},\bm{\mu}\setminus\bm{\nu}} & 0 & 0 & 0 \\
		\hline
		0 & Q_{\bm{\mu}\cap\bm{\nu},\bm{\mu}\cap\bm{\nu}} & 0 & 0 \\
		\hline
		0 & 0 & Q_{\bm{\nu}\setminus\bm{\mu},\bm{\nu}\setminus\bm{\mu}} & 0 \\
		\hline
		0 & 0 & 0 & *
		\end{array} \r),
	\end{equation}
	where the linear sizes of the four blocks are $ k-j $, $ j $, $ k-j $, and $ 2(n-k)+j $, respectively. If either $ \det[ Q_{\bm{\mu},\bm{\mu}} ] $ or $ \det[ Q_{\bm{\nu},\bm{\nu}} ] $ are $ 0 $, then that term in the sum trivially vanishes. Thus consider the case in which they are both nonzero. Again, since we have the constraint $ \det Q = 1 $, the product of the determinants of all four blocks must be $ 1 $. There are eight cases in which $ \det[ Q_{\bm{\mu},\bm{\mu}} ] \det[ Q_{\bm{\nu},\bm{\nu}} ] \neq 0 $:~$ Q \in \Sym^a(2,k-j) \oplus \Sym^b(2,j) \oplus \Sym^c(2,k-j) \oplus \Sym^d(2,2(n-k)+j) $, where
	\begin{equation}\label{eq:det_signs}
		(a,b,c,d) \in
		\l\{ \begin{array}{c c c c}
			({+},{+},{+},{+}), &
			({+},{+},{-},{-}), &
			({+},{-},{+},{-}), &
			({+},{-},{-},{+}), \\
			({-},{+},{+},{-}), &
			({-},{+},{-},{+}), &
			({-},{-},{+},{+}), &
			({-},{-},{-},{-}) 
		\end{array} \r\}.
	\end{equation}
	With this formalism, we can read off the terms in the off-diagonal sum as $ \det[ Q_{\bm{\mu},\bm{\mu}} ] \det[ Q_{\bm{\nu},\bm{\nu}} ] = (ab)(bc) = ac $ (where our notation means $ a = \pm \equiv \pm 1 $, etc.). By examining \cref{eq:det_signs}, we see that four of the possibilities give $ ac = +1 $, while the other four give $ ac = -1 $. Since the number of elements in each of the eight subsets are all the same, exactly half the terms in the sum will cancel with the other half, hence
	\begin{equation}
		\sum_{Q \in G} \det[ Q_{\bm{\mu},\bm{\mu}} ] \det[ Q_{\bm{\nu},\bm{\nu}} ] = 0 \quad \forall \bm{\mu} \neq \bm{\nu}
	\end{equation}
	as desired.
	
	Thus,
	\begin{equation}
		\frac{1}{|G|} \sum_{Q \in G} \tr[ \phi_k(Q) ]^2 = 1,
	\end{equation}
	and so $ \phi_k\big|_G $ is irreducible.
	
	The proof readily adapts for $ G = \Alt(2n) $. As we saw in \cref{eq:avg_density_calc}, the fact that the matrix elements have signs is irrelevant, as the factors arising due to the wreath product (i.e., $ 2^k $, $ 2^{2n-k} $, and $ 2^{2n} $) exactly cancel out. One may then simply replace every appearance of $ \Sym^{+}(2,d) $ with $ \Sym^{+}(d) \equiv \Alt(d) $ and $ \Sym^{-}(2,d) $ with $ \Sym^{-}(d) $ without consequence.
\end{proof}

In conjunction with the rigorous guarantees of classical shadows, \cref{thm:fgu_performance} of the main text follows. To simplify the shadow norm expression, we use Stirling's approximation, yielding
\begin{equation}
    \left. \binom{2n}{2k} \middle/ \binom{n}{k} \right. \approx \binom{n}{k} \sqrt{\pi k}.
\end{equation}
Additionally, the $ \log L $ factor from using classical shadows ($ L $ being the number of Majorana operators here) is
\begin{equation}
	\log\l( \sum_{j=1}^k |\comb{2n}{2j}| \r) = \O(k \log n).
\end{equation}

\subsection{\label{sec:class_shadow_appendix}The classical shadow estimator}

Here we provide a derivation for the formal expressions of the linear-inversion estimator used in the classical shadows methodology. Recalling \cref{eq:UzzU}, and using the linearity of $ \mathcal{M}^{-1}_{\FGU} $, the classical shadow is simply
\begin{equation}
    \begin{split}
    \hat{\rho}_{Q,z} &\equiv \mathcal{M}^{-1}_{\FGU} \l( U(Q)^\dagger \op{z}{z} U(Q) \r)\\
    &= \frac{1}{2^n} \l( \openone + \sum_{j=1}^n \lambda_{n,j}^{-1} \sum_{\bm{\mu} \in \diags{2n}{2j}} \ev{z}{\Gamma_{\bm{\mu}}}{z} \sum_{\bm{\nu} \in \comb{2n}{2j}} \det[Q_{\bm{\mu},\bm{\nu}}] \Gamma_{\bm{\nu}} \r),
    \end{split}
\end{equation}
where $ \lambda_{n,j}^{-1} = \binom{2n}{2j}/\binom{n}{j} $. Passing this expression into the expectation value estimator, we then obtain
\begin{equation}\label{eq:ev_estimator_app}
    \begin{split}
	\tr(\Gamma_{\bm{\tau}} \hat{\rho}_{Q,z}) &= \frac{1}{2^n} \l( \tr(\Gamma_{\bm{\tau}}) + \sum_{j=1}^n \lambda_{n,j}^{-1} \sum_{\bm{\mu} \in \diags{2n}{2j}} \ev{z}{\Gamma_{\bm{\mu}}}{z} \sum_{\bm{\nu} \in \comb{2n}{2j}} \det[Q_{\bm{\mu},\bm{\nu}}] \tr(\Gamma_{\bm{\tau}} \Gamma_{\bm{\nu}}) \r)\\
	&= \lambda_{n,k}^{-1} \sum_{\bm{\mu} \in \diags{2n}{2k}} \ev{z}{\Gamma_{\bm{\mu}}}{z} \det[Q_{\bm{\mu},\bm{\tau}}]
	\end{split}
\end{equation}
for all $ \bm{\tau} \in \comb{2n}{2k} $.

\subsection{\label{sec:arb_observable_appendix}Variance bounds for arbitrary observables}

For completeness, we provide a bound on the shadow norm (and hence the estimator variance) of an arbitrary fermionic observable. This result is particularly useful in the context of Hamiltonian averaging (e.g., \cref{sec:hamiltonian_appendix}).

Any element of $ \MA{n}{\mathrm{even}} $ can be written as
\begin{equation}
	O = h_{\varnothing} \openone + \sum_{j=1}^n \sum_{\bm{\mu} \in \comb{2n}{2j}} h_{\bm{\mu}} \Gamma_{\bm{\mu}},
\end{equation}
where $ h_{\bm{\mu}} \in \R $. Without loss of generality, we shall take $ \tr O = 0 $, since the identity component is irrelevant for variance calculations. Furthermore, to simplify the following exposition we shall suppose that $ O $ is at most a $ k $-body operator, so that $ h_{\bm{\mu}} = 0 $ for all $ |\bm{\mu}| > 2k $:
\begin{equation}
	O = \sum_{j=1}^k \sum_{\bm{\mu} \in \comb{2n}{2j}} h_{\bm{\mu}} \Gamma_{\bm{\mu}}.
\end{equation}
While we do not impose any restriction on $ k $, it is worth noting that most physical observables of interest obey $ k \leq 4 $, with $ k = 2 $ being particularly important (for instance, in describing electron--electron interactions).

Because the shadow norm is indeed a norm, it obeys the triangle inequality~\cite{huang2020predicting}. This property allows us to place a bound on
\begin{equation}
    \begin{split}
    \sns{O}{\FGU} &= \l\| \sum_{j=1}^k \sum_{\bm{\mu} \in \comb{2n}{2j}} h_{\bm{\mu}} \Gamma_{\bm{\mu}} \r{\|_{\FGU}}\\
    &\leq \sum_{j=1}^k \sum_{\bm{\mu} \in \comb{2n}{2j}} |h_{\bm{\mu}}| \sns{\Gamma_{\bm{\mu}}}{\FGU}\\
    &= \sum_{j=1}^k \lambda_{n,j}^{-1/2} \sum_{\bm{\mu} \in \comb{2n}{2j}} |h_{\bm{\mu}}|.
    \end{split}
\end{equation}
We therefore obtain an upper bound on the variance of our classical shadow estimator for $ \tr(O \rho) $ as
\begin{equation}
    \begin{split}
    \V_{Q,z} \l[ \tr(O \hat{\rho}_{Q,z}) \r] &\leq \sns{O}{\FGU}^2 - \tr(O \rho)^2\\
    &\leq \l( \sum_{j=1}^k \lambda_{n,j}^{-1/2} \sum_{\bm{\mu} \in \comb{2n}{2j}} |h_{\bm{\mu}}| \r{)^2} - \tr(O \rho)^2.
    \end{split}
\end{equation}
To get a sense for the asymptotic scaling of this expression, one may further loosen the estimate to obtain
\begin{equation}
    \V_{Q,z} \l[ \tr(O \hat{\rho}_{Q,z}) \r] \leq \l( \max_{1 \leq \ell \leq k} \lambda_{n,\ell}^{-1} \r) \l( \sum_{j=1}^k \sum_{\bm{\mu} \in \comb{2n}{2j}} |h_{\bm{\mu}}| \r{)^2},
\end{equation}
where $ \max_{1 \leq \ell \leq k} \lambda_{n,\ell}^{-1} = \O(n^k) $ when $ k = \O(1) $.

In follow-up works to this paper, exact expressions for the variance and tighter bounds were derived by Wan \emph{et al.}~\cite{wan2022matchgate} and O'Gorman~\cite{ogorman2022fermionic}.

\subsection{\label{sec:hoeffding}Performance guarantees without median-of-means estimation}

As remarked in the main text, we do not require the median-of-means technique proposed in the original work~\cite{huang2020predicting} to obtain the same rigorous sampling bounds. Instead, one may simply use the typical sample mean:~given $ M $ independently obtained classical shadows $ \hat{\rho}_1, \ldots, \hat{\rho}_M $, define
\begin{equation}
    \omega_j(M) \coloneqq \frac{1}{M} \sum_{i=1}^M \tr(O_j \hat{\rho}_i)
\end{equation}
for each $ j \in \{1,\ldots,L\} $. Below, we state a general condition for which this estimator yields sample complexity equivalent to that of the median-of-means estimator.

\begin{theorem}\label{thm:sample_mean_estimator}
    Suppose the classical shadow estimators $ \hat{\rho}_{U,z} $ satisfy
    \begin{equation}\label{eq:rv_bound}
        -\sns{O_j}{\mathcal{U}}^2 \leq \tr(O_j \hat{\rho}_{U,z}) \leq \sns{O_j}{\mathcal{U}}^2
    \end{equation}
    for all $ U \in \mathcal{U} $, $ z \in \{0,1\}^n $, and $ j \in \{1,\ldots,L\} $. Let $ \varepsilon, \delta \in (0,1) $. Then by setting
    \begin{equation}\label{eq:sample_mean_M}
        M = \l(1 + \frac{\varepsilon}{3}\r) \frac{2 \log(2L/\delta)}{\varepsilon^2} \max_{1 \leq j \leq L} \sns{O_j}{\mathcal{U}}^2,
    \end{equation}
    we ensure that all sample-mean estimators $ \omega_1(M), \ldots, \omega_L(M) $ satisfy
    \begin{equation}
        | \omega_j(M) - \tr(O_j \rho) | \leq \varepsilon,
    \end{equation}
    with probability at least $ 1 - \delta $.
\end{theorem}

\begin{proof}
    The claim follows straightforwardly from Bernstein's inequality~\cite[Eq.~(2.10)]{boucheron2013concentration}:~for a collection of independent random variables $ X_1, \ldots, X_M $ satisfying $ |X_i| \leq b $ for all $ i \in \{1,\ldots,M\} $, the probability that their empirical mean $ \bar{X} \coloneqq \frac{1}{M}\sum_{i=1}^M X_i $ deviates from the true mean $ \E[\bar{X}] $ by more than $ \varepsilon $ is bounded as
    \begin{equation}
        \Pr\l[ | \bar{X} - \E[\bar{X}] | \geq \varepsilon \r] \leq 2 \exp\l( - \frac{M^2 \varepsilon^2 / 2}{v + bM\varepsilon/3} \r),
    \end{equation}
    where $ v \coloneqq \sum_{i=1}^M \E[X_i^2] $.
    
    In our setting, for each $ j \in \{1,\ldots,L\} $, we have $ \bar{X} = \omega_j(M) $, $ b = \sns{O_j}{\mathcal{U}}^2 $, and $ v = M \sns{O_j}{\mathcal{U}}^2 $ (recall that the shadow norm squared is precisely $ \E[X_i^2] $~\cite{huang2020predicting}). The concentration inequality then reads
    \begin{equation}
        \Pr\l[ | \omega_j(M) - \tr(O_j \rho) | \geq \varepsilon \r] \leq 2 \exp\l[ - \frac{M \varepsilon^2 / 2}{\sns{O_j}{\mathcal{U}}^2 (1 + \varepsilon/3)} \r].
    \end{equation}
    If we require that each probability of failure be no more than $ \delta/L $, then from a union bound over all $ L $ events, we can succeed with probability at least $ 1 - \delta $ by setting
    \begin{equation}
        2 \exp\l( - \frac{M \varepsilon^2 / 2}{\max_{1 \leq j \leq L} \sns{O_j}{\mathcal{U}}^2 (1 + \varepsilon/3)} \r) = \frac{\delta}{L}.
    \end{equation}
    Solving for $ M $ yields \cref{eq:sample_mean_M}.
\end{proof}

For practical purposes, one typically desires that $ \varepsilon $ be small, thus $ (1 + \varepsilon/3) \approx 1 $. Importantly, \cref{thm:sample_mean_estimator} guarantees optimal scaling with the failure probability $ \delta $ and the number $ L $ of observables, using the sample mean rather than median-of-means estimation. The key detail which enables this observation is the boundedness of the classical shadows estimators, \cref{eq:rv_bound}. This condition holds for the ensembles presented in this work, $ \mathcal{U} = \mathcal{U}_{\FGU} $ and $ \mathcal{U}_{\NC} $ (see \cref{sec:ncu_appendix}), when taking the observables $ O_j $ as Majorana operators. It is also satisfied for estimating Pauli observables using the $ \Cl(1)^{\otimes n} $ ensemble of the original work~\cite{huang2020predicting}, indicating that median-of-means is redundant for estimating qubit RDMs as well.

Interestingly, \cref{eq:sample_mean_M} features significantly smaller numerical factors than what is obtained from the median-of-means approach, although we recognize that the proof techniques presented in Ref.~\cite{huang2020predicting} were not particularly optimized in this regard.

\section{\label{sec:ncu_appendix}Computations with the number-conserving modification}

We now prove \cref{thm:NC_bounds} of the main text. Many of the techniques used here follow straightforwardly from the fermionic formalism developed in \cref{sec:fgu_appendix}, along with the tools used to study the single-qubit Clifford ensemble in the original work on classical shadows~\cite{huang2020predicting}. In \cref{sec:ncu_calcs}, we first evaluate the channel $ \mathcal{M}_\NC $ [which again is diagonalized by the Majorana operators;~\cref{eq:NC_eig}] and provide an expression for its eigenvalues/shadow norm [\cref{eq:nc_norm}]. Then, recognizing that generically do not possess a closed-form expression, in \cref{sec:ncu_ub} we obtain an upper bound on the shadow norm for this ensemble, which is asymptotically optimal [\cref{eq:nc_norm_bound}].

For ease of notation, with $ u \in \Sym(n) $, we shall write $ U(u)^\dagger a_p U(u) = a_{u(p)} $, where $ u(p) $ is understood in the sense of the action of the permutation $ u $ on the mode indices $ \{ 0,\ldots,n-1 \} $. We further generalize this notation to act on $ \{ 0,\ldots,2n-1 \} $, in accordance with the definition of Majorana operators. Consider $ \bm{\mu} \in \comb{2n}{k} $, where each index takes the form $ \mu_j = 2q_j + x_j $ for $ q_j \in \{ 0,2,\ldots,2n-4,2n-2 \} $ and $ x_j \in \{ 0,1 \} $. We define $ \tilde{u}(\bm{\mu}) \coloneqq (2u(q_1)+x_1, \ldots, 2u(q_{k})+x_{k} ) $, so that $ U(u)^\dagger \Gamma_{\bm{\mu}} U(u) = \Gamma_{\tilde{u}(\bm{\mu})} $. Note that $ \tilde{u}(\bm{\mu}) $ is not necessarily ordered monotonically, so $ \Gamma_{\tilde{u}(\bm{\mu})} $ may differ from our standard definition of the Majorana operators by a minus sign. This detail is irrelevant to our present analysis, so we shall ignore it.

\subsection{\label{sec:ncu_calcs}The shadow norm}

The ensemble we consider here is
\begin{equation}
    \mathcal{U}_{\mathrm{NC}} = \{ V \circ U(u) \mid V \in \Cl(1)^{\otimes n}, \,  u \in \Alt(n) \}.
\end{equation}
We wish to evaluate the following expression for the classical shadows channel:
\begin{equation}
    \begin{split}
	\mathcal{M}_{\mathrm{NC}}(\Gamma_{\bm{\mu}}) &= \E_{\substack{V \sim \Cl(1)^{\otimes n} \\ u \sim \Alt(n)}} \l[ \sum_{z\in\{0,1\}^n} \ev{z}{V U(u) \Gamma_{\bm{\mu}} U(u)^\dagger V^\dagger}{z} U(u)^\dagger V^\dagger \op{z}{z} V U(u) \r]\\
	&= \E_{u \sim \Alt(n)} \l[ U(u)^\dagger \sum_{z\in\{0,1\}^n} \E_{V \sim \Cl(1)^{\otimes n}} \Big[ \ev{z}{V \Gamma_{\tilde{u}^{-1}(\bm{\mu})} V^\dagger}{z} V^\dagger \op{z}{z} V \Big] U(u) \r].
	\end{split}
\end{equation}
The average over $ \Cl(1)^{\otimes n} $ is precisely the same quantity evaluated in Ref.~\cite{huang2020predicting}. For convenience, we restate their results here:~let $ \ket{x} = \ket{x_1} \otimes \cdots \otimes \ket{x_n} $ be a product state over $ n $ qubits and $ A,B,C $ be Hermitian matrices which decompose into the same tensor product structure, $ A = A_1 \otimes \cdots \otimes A_n $, etc. Then
\begin{equation}\label{eq:cl_avg_1}
    \E_{V \sim \Cl(1)^{\otimes n}} \l[ V^\dagger \op{x}{x} V \ev{x}{V A V^\dagger}{x} \r] = \bigotimes_{j=1}^n \l( \frac{A_j + \tr(A_j) I}{6} \r),
\end{equation}
and for each $ j $ such that $ \tr B_j = \tr C_j = 0 $,
\begin{equation}\label{eq:cl_avg_2}
    \E_{V_j \sim \Cl(1)} \l[ V_j^\dagger \op{x_j}{x_j} V_j \ev{x_j}{V_j B_j V_j^\dagger}{x_j} \ev{x_j}{V_j C_j V_j^\dagger}{x_j} \r] = \frac{B_j C_j + C_j B_j + \tr(B_j C_j) I}{24}.
\end{equation}
To understand the tensor product structure of Majorana operators, we must fix some qubit mapping. Understanding $ \Gamma_{\tilde{u}^{-1}(\bm{\mu})} $ as a Pauli operator under such a mapping, we use \cref{eq:cl_avg_1} to obtain
\begin{equation}\label{eq:NC_eig}
    \begin{split}
    \mathcal{M}_{\mathrm{NC}}(\Gamma_{\bm{\mu}}) &= \E_{u \sim \Alt(n)} \l[ \frac{1}{3^{\loc(\Gamma_{\tilde{u}^{-1}(\bm{\mu})})}} U(u)^\dagger \Gamma_{\tilde{u}^{-1}(\bm{\mu})} U(u) \r]\\
    &= \E_{u \sim \Alt(n)} \l[ \frac{1}{3^{\loc(\Gamma_{\tilde{u}^{-1}(\bm{\mu})})}} \r] \Gamma_{\bm{\mu}},
    \end{split}
\end{equation}
where $ \loc(\Gamma_{\tilde{u}^{-1}(\bm{\mu})}) $ is the qubit locality of $ \Gamma_{\tilde{u}^{-1}(\bm{\mu})} $. Since $ u \mapsto u^{-1} $ is a bijection, we may equivalently express the average over $ \Alt(n) $ using $ \tilde{u}(\bm{\mu}) $ rather than its inverse.

Let $ \lambda_{\bm{\mu}} $ be the eigenvalues of $ \mathcal{M}_{\mathrm{NC}} $, given above in \cref{eq:NC_eig}. For the calculation of the shadow norm, we have
\begin{align}
	\sns{\Gamma_{\bm{\mu}}}{\mathrm{NC}}^2 &= \max_{\text{states } \rho} \l( \E_{\substack{V \sim \Cl(1)^{\otimes n} \\ u \sim \Alt(n)}} \l[ \sum_{z\in\{0,1\}^n} \ev{z}{V U(u) \rho U(u)^\dagger V^\dagger}{z} \ev{z}{V U(u) \mathcal{M}^{-1}(\Gamma_{\bm{\mu}}) U(u)^\dagger V^\dagger}{z}^2 \r] \r),\\
	&= \lambda_{\bm{\mu}}^{-2} \max_{\text{states } \rho} \l( \E_{u \sim \Alt(n)} \l[ \tr\l( U(u)^\dagger \rho U(u) \sum_{z\in\{0,1\}^n} \E_{V \sim \Cl(1)^{\otimes n}} \Big[ V^\dagger \op{z}{z} V \ev{z}{V \Gamma_{\tilde{u}^{-1}(\bm{\mu})} V^\dagger}{z}^2 \Big] \r) \r] \r). \notag
\end{align}
We then use \cref{eq:cl_avg_1} to evaluate the Clifford average over the identity factors of $ \Gamma_{\tilde{u}^{-1}(\bm{\mu})} $, and \cref{eq:cl_avg_2} for the nontrivial factors:
\begin{equation}\label{eq:nc_norm}
    \begin{split}
    \sns{\Gamma_{\bm{\mu}}}{\mathrm{NC}}^2 &= \lambda_{\bm{\mu}}^{-2} \max_{\text{states } \rho} \l( \E_{u \sim \Alt(n)} \l[ \tr\l( U(u)^\dagger \rho U(u) \frac{1}{3^{\loc(\Gamma_{\tilde{u}^{-1}(\bm{\mu})})}} \r) \r] \r)\\
    &= \lambda_{\bm{\mu}}^{-1}.
    \end{split}
\end{equation}

\subsection{\label{sec:ncu_ub}Universal upper bounds on the shadow norm}

Although there is no closed-form expression for $ \lambda_{\bm{\mu}}^{-1} $, we can still obtain a nontrivial estimate for it. To do so, we will evaluate the qubit locality with respect to the Jordan--Wigner transformation. This serves as a universal upper bound for all encodings, since the Jordan--Wigner mapping is maximally nonlocal. We formalize this notion with Jensen's inequality:~since $ x \mapsto 3^{-x} $ is concave, we have that
\begin{equation}
	\E_{u \sim \Alt(n)} \l[ 3^{-\loc(\Gamma_{\tilde{u}(\bm{\mu})})} \r] \geq 3^{-\E_{u \sim \Alt(n)}\l[ \loc(\Gamma_{\tilde{u}(\bm{\mu})}) \r]},
\end{equation}
where $ \loc(\,\cdot\,) $ is with respect to any arbitrary encoding. Then let $ \loc_{\mathrm{JW}}(\,\cdot\,) $ be the qubit locality specifically under the Jordan--Wigner transformation. Since it is maximally nonlocal, so is its average locality (over a fixed fermionic degree of $ 2k $), and hence
\begin{equation}
	\max_{\bm{\mu} \in \comb{2n}{2k}} \lambda_{\bm{\mu}}^{-1} \leq \max_{\bm{\mu} \in \comb{2n}{2k}} 3^{\E_{u \sim \Alt(n)}\l[ \loc(\Gamma_{\tilde{u}(\bm{\mu})}) \r]} \leq \max_{\bm{\mu} \in \comb{2n}{2k}} 3^{\E_{u \sim \Alt(n)}\l[ \loc_{\mathrm{JW}}(\Gamma_{\tilde{u}(\bm{\mu})}) \r]}.
\end{equation}
Thus for the rest of this section, all notions of qubit locality will be understood with respect to the Jordan--Wigner transformation exclusively.

Since the qubit locality of Majorana operators varies within a given $ \comb{2n}{2k} $, we obtain an upper bound by considering the most nonlocal $ 2k $-degree operators. Fortunately, under Jordan--Wigner it is simple to identify such operators. Let $ W \in \{ X,Y \} $ and $ \bm{q} \in \comb{n}{2k} $;~maximum locality is achieved by operators of the form
\begin{align}
	P_{\bm{q}} &\coloneqq \prod_{i=1}^k W_{q_{2i-1}} Z_{q_{2i-1}+1} \cdots Z_{q_{2i}-1} W_{q_{2i}},\\
	\loc(P_{\bm{q}}) &= \sum_{i=1}^k \l( q_{2i} - q_{2i-1} + 1 \r).
\end{align}
Note that we require $ n \geq 2k $ for such a Pauli operator to exist. (For the fringe cases in which $ n < 2k $, since we are evaluating an upper bound, the final result still holds.)

By applying a permutation $ u $ on $ \bm{q} $, one changes this locality by reducing or lengthening the various ``Jordan--Wigner strings'' of Pauli-$ Z $ operators in between each pair $ (q_{2i-1}, q_{2i}) $. Observe that $ \loc(P_{\bm{q}}) \in \{ 2k, \ldots, n \} $. Therefore, if we can calculate the number of permutations which correspond to each level of locality, we may compute
\begin{equation}\label{eq:ncsq_norm_comb}
	\E_{u \sim \Alt(n)} \l[ 3^{-\loc(P_{u(\bm{q})})} \r] = \frac{2}{n!} \sum_{\ell=2k}^n 3^{-\ell} \, | \{ u \in \Alt(n) \mid \loc(P_{u(\bm{q})}) = \ell \} |.
\end{equation}
We determine the size of this set in three steps. First, for a given configuration $ u(\bm{q}) $, we count the number of equivalent permutations which give the same Pauli operator, modulo signs. These permutations merely reorder the $ 2k $ relevant indices and the $ n-2k $ remaining indices independently. Accounting for the fact that the entire permutation must be even parity, we have
\begin{equation}
	\begin{split}
	|\Alt(2k) \oplus \Alt(n - 2k)| + |\Sym^{-}(2k) \oplus \Sym^{-}(n - 2k)| &= \frac{(2k)!}{2} \frac{(n-2k)!}{2} + \frac{(2k)!}{2} \frac{(n-2k)!}{2}\\
	&= \frac{(2k)!(n-2k)!}{2}
	\end{split}
\end{equation}
such permutations. With this factor at hand, we now only need to consider unique combinations of indices---that is, we assume $ u(q_1) < \cdots < u(q_{2k}) $ in the sequel.

Next, we calculate how many configurations of Jordan--Wigner strings give rise to a qubit locality of exactly $ \ell $. Let $ l_i \coloneqq u(q_{2i}) - u(q_{2i-1}) > 0 $. This problem is equivalent to finding all $ k $-tuples $ (l_1, \ldots, l_k) $ of positive integers such that
\begin{equation}\label{eq:snb_constraint}
	\sum_{i=1}^k (l_i + 1) = \ell.
\end{equation}
This is an instance of the classic ``stars-and-bars'' combinatorial problem, wherein we wish to fit $ \ell - k $ objects into $ k $ bins such that each bin has at least $ 1 $ object. There are $ \binom{\ell-k-1}{k-1} $ ways to do so.

Finally, we need to account for the $ n-\ell $ remaining indices which were not fixed by \cref{eq:snb_constraint}. These indices correspond to the qubits \emph{outside} of the Jordan--Wigner strings, i.e., on which $ P_{u(\bm{q})} $ acts trivially. By a similar combinatorial argument, there are a total of $ k+1 $ such spaces to place these trivial indices, of which we may select $ j \in \{ 1,\ldots,k+1 \} $. Again we use the stars-and-bars argument:~there are $ n-\ell $ objects we wish to place into $ j $ bins, which can be accomplished in $ \binom{n-\ell-1}{j-1} $ unique ways. Summing over all possible values of $ j $ gives us
\begin{equation}
	\begin{split}
	\sum_{j=1}^{k+1} \binom{k+1}{j} \binom{n-\ell-1}{j-1} &= \sum_{j=1}^{k+1} \binom{k+1}{k+1-j} \binom{n-\ell-1}{j-1}\\
	&= \sum_{j=0}^{k} \binom{k+1}{k-j} \binom{n-\ell-1}{j} = \binom{n+k-\ell}{k}
	\end{split}
\end{equation}
different combinations of the $ n - \ell $ trivial indices. Note that the Chu--Vandermonde identity was applied to evaluate the sum in the final line.

Reconciling the three steps of our calculation, we obtain
\begin{equation}
	| \{ u \in \Sym(n) \mid \loc(P_{u(\bm{q})}) = \ell \} | = (2k)!(n-2k)! \binom{\ell-k-1}{k-1} \binom{n+k-\ell}{k},
\end{equation}
and so \cref{eq:ncsq_norm_comb} can be expressed as
\begin{equation}
	\begin{split}
	\E_{u \sim \Sym(n)} \l[ 3^{-\loc(P_{u(\bm{q})})} \r] &= \frac{1}{n!} \sum_{\ell=2k}^n 3^{-\ell} \, | \{ u \in \Sym(n) \mid \loc(P_{u(\bm{q})}) = \ell \} |\\
	&= \frac{(2k)!(n-2k)!}{n!} \sum_{\ell=2k}^n 3^{-\ell} \binom{\ell-k-1}{k-1} \binom{n+k-\ell}{k}.
	\end{split}
\end{equation}
To evaluate this sum, first we relabel the index $ \ell $ to run from $ 0 $ to $ n-2k $, so that the summand becomes $ 3^{-(\ell + 2k)} \binom{\ell+k-1}{k-1} \binom{n-k-\ell}{k} $. Using the combinatorial identity
\begin{equation}
    \binom{n-k}{\ell} \binom{n-k-\ell}{k} = \binom{n-k}{k} \binom{n-2k}{\ell},
\end{equation}
we have
\begin{equation}
    \begin{split}
    \binom{n-k-\ell}{k} &= \binom{n-k}{k} \frac{(n-2k)!}{(n-2k-\ell)! \ell!} \frac{(n-k-\ell)! \ell!}{ (n-k)! }\\
    &= \binom{n-k}{k} \frac{(2k-n)_\ell}{(k-n)_\ell},
    \end{split}
\end{equation}
where $ (m)_\ell \coloneqq \prod_{j=0}^{n-1} (m+j) $ is the rising Pochhammer symbol. We recognize that the other binomial coefficient can also be expressed using these Pochhammer symbols:
\begin{equation}
    \binom{\ell+k-1}{k-1} = \frac{(k)_\ell}{\ell!}.
\end{equation}
The sum can therefore be understood in terms of the Gauss hypergeometric function $ {}_2 F_1 $:
\begin{equation}\label{eq:avg_f}
    \begin{split}
    \E_{u \sim \Sym(n)} \l[ 3^{-\loc(P_{u(\bm{q})})} \r] &= \binom{n}{2k}^{-1} \binom{n-k}{k} 9^{-k} \sum_{\ell=0}^{n-2k} \frac{3^{-\ell}}{\ell!} \frac{(k)_\ell (2k-n)_\ell}{(k-n)_\ell}\\
    &= \binom{n}{2k}^{-1} \binom{n-k}{k} 9^{-k} \, {}_2 F_1(k, 2k-n; k-n; 1/3).
    \end{split}
\end{equation}
Using standard properties of the hypergeometric function, for $ n $ and $ k $ being positive integers and $ k $ a constant, we have the bounds $ 1 \leq {}_2 F_1(k, 2k-n; k-n; 1/3) \leq (3/2)^k $. In particular, this factor is independent of $ n $.

We are ultimately interested in the lower bound of \cref{eq:avg_f}, since the shadow norm is its reciprocal. Then for $ k = \O(1) $, we obtain
\begin{equation}\label{eq:nc_norm_bound}
	\sns{P_{\bm{q}}}{\mathrm{NC}}^2 \leq \binom{n}{2k} \binom{n-k}{k}^{-1} 9^k = \O(n^k).
\end{equation}

\section{\label{sec:rdmswap_appendix}Fermionic swap network bounds}

In this section we describe another strategy to measure the $ k $-RDM with almost optimal scaling in $ n $. Generalizing from the measurement scheme introduced in Ref.~\cite{arute2020hartree} for 1-RDMs, this scheme employs fermionic swap gates to relabel which qubits correspond to which orbitals such that each $ k $-RDM observable becomes a $ 2k $-local qubit observable. These local qubit observables may then be measured in a parallel fashion via Pauli measurements. As the circuit structure is equivalent, our $ \mathcal{U}_{\NC} $-based scheme may be viewed as a randomized version of this strategy.

For simplicity, we employ the Jordan--Wigner encoding throughout this section. The use of fermionic swap networks to minimize the qubit locality of fermionic operators was first utilized in the context of Hamiltonian simulation~\cite{kivlichan2018quantum}.

\subsection{The 1-RDM method}

We briefly describe the methods of Ref.~\cite{arute2020hartree} here. Consider estimating the 1-RDM elements $ \tr\l( a_p^\dagger a_q \rho \r) $. The observables here are $ \frac{1}{2}(a_p^\dagger a_q + \mathrm{h.c.}) $ and $ \frac{1}{2i}(a_p^\dagger a_q - \mathrm{h.c.}) $, corresponding to the real and imaginary parts of the RDM element. In the experiment of Ref.~\cite{arute2020hartree}, they implement unitaries such that the imaginary part vanishes;~for full generality, we will keep the imaginary parts. Diagonal elements are trivial to measure, since
\begin{equation}
	a_p^\dagger a_p = \frac{I - Z_p}{2}.
\end{equation}
The one-off-diagonal terms are precisely the local qubit operators we are interested in:
\begin{align}
	\frac{a_{p}^\dagger a_{p+1} + a_{p+1}^\dagger a_{p}}{2} &= \frac{X_{p} X_{p+1} + Y_{p} Y_{p+1}}{4},\\
	\frac{a_{p}^\dagger a_{p+1} - a_{p+1}^\dagger a_{p}}{2i} &= \frac{X_{p} Y_{p+1} - Y_{p} X_{p+1}}{4}.
\end{align}
Consider even and odd pairs of orbitals---even pairs being those starting with even indices, and analogously for the odd pairs. For the expectation values $ \tr\l( a_p^\dagger a_{p+1} \rho \r) $ on even pairs, we measure in 4 different bases:~$ X $ on all qubits, $ Y $ on all qubits, $ X $ on every even qubit and $ Y $ on every odd qubit, and vice versa. Formally, the observables we measure are
\begin{equation}
    \begin{split}
	O_1 = \prod_{p=0}^{n-1} X_p, \quad & O_2 = \prod_{p=0}^{n-1} Y_p\\
	O_3 = \prod_{\substack{p=0\\\text{even}}}^{n-2} X_p Y_{p+1}, \quad & O_4 = \prod_{\substack{p=0\\\text{even}}}^{n-2} Y_p X_{p+1}.
	\end{split}
\end{equation}
If $ n $ is odd, then we simply ignore the final $ Y $ (resp.~$ X $) in $ O_3 $ (resp.~$ O_4 $).

The remaining off-diagonal elements will incur Jordan--Wigner strings, making the Pauli operators highly nonlocal. To circumvent this, we perform fermionic swaps to relabel the indices such that we retain qubit locality. The fermionic swap gate between orbitals $ p $ and $ q $ is
\begin{equation}
	\mathcal{F}_{pq} = \exp\l[ -i \frac{\pi}{2} ( a_p^\dagger a_q + a_q^\dagger a_p - a_p^\dagger a_p - a_q^\dagger a_q ) \r].
\end{equation}
This unitary is Gaussian and number-preserving, hence one may use the group homomorphism property to consolidate an arbitrarily large product of fermionic swaps into a single circuit of depth $ n $~\cite{kivlichan2018quantum,jiang2018quantum}.

As described in Ref.~\cite{kivlichan2018quantum}, a total of $ \lceil n/2 \rceil $ different swap circuits are required to move orbitals such that every pair is nearest-neighbor at least once. This quantity can be understood from a simple counting argument:~there are $ \binom{n}{2} $ off-diagonal 1-RDM elements (orbital pairs) to account for. Each ordering creates $ n-1 $ nearest-neighbor pairs---therefore, we require $ \binom{n}{2}/(n-1) = n/2 $ different orderings (hence unique swap circuits) to match all pairs of orbitals. In practice, this is achieved using a parallelized odd--even transposition sort.

Rounding up in the case that $ n $ is odd, and accounting for the four Pauli bases per permutation, the total number of different measurement circuits is $ 4 \lceil n/2 \rceil + 1 $.

\subsection{The 2-RDM method}

We now generalize the use of such swap networks for measuring $ k $-RDMs. We will build up intuition with the 2-RDM. Diagonal terms are trivial as usual, since
\begin{equation}
	a_p^\dagger a_q^\dagger a_q a_p = (1 - \delta_{pq}) \frac{I - Z_p - Z_q + Z_p Z_q}{4}.
\end{equation}
The terms with a single occupation-number operator, restricted to 3-qubit locality, are
\begin{align}
	\frac{a_p^\dagger a_{q}^\dagger a_{q} a_{p+1} + \mathrm{h.c.}}{2} &= \frac{ X_p X_{p+1} + Y_p Y_{p+1} - X_p X_{p+1} Z_{q} - Y_p Y_{p+1} Z_{q} }{8},\\
	\frac{a_p^\dagger a_{q}^\dagger a_{q} a_{p+1} - \mathrm{h.c.}}{2i} &= \frac{ X_p Y_{p+1} - Y_p X_{p+1} - X_p Y_{p+1} Z_{q} + Y_p X_{p+1} Z_{q} }{8},
\end{align}
where $ q \notin \{ p,p+1 \} $. Lastly, we have the most general case:
\begin{align}\label{eq:general_tpdm}
	\frac{ a_{p}^\dagger a_{p+1}^\dagger a_{q} a_{q+1} + \mathrm{h.c.} }{2} &= \frac{1}{16} ( - X_{p} X_{p+1} X_{q} X_{q+1} + X_{p} X_{p+1} Y_{q} Y_{q+1} - X_{p} Y_{p+1} X_{q} Y_{q+1} - X_{p} Y_{p+1} Y_{q} X_{q+1} \\
	&\qquad\ \ \, - Y_{p} X_{p+1} X_{q} Y_{q+1} - Y_{p} X_{p+1} Y_{q} X_{q+1} + Y_{p} Y_{p+1} X_{q} X_{q+1} - Y_{p} Y_{p+1} Y_{q} Y_{q+1} ), \notag\\
	\frac{ a_{p}^\dagger a_{p+1}^\dagger a_{q} a_{q+1} - \mathrm{h.c.} }{2i} &= \frac{1}{16} ( - X_{p} X_{p+1} X_{q} Y_{q+1} - X_{p} X_{p+1} Y_{q} X_{q+1} + X_{p} Y_{p+1} X_{q} X_{q+1} - X_{p} Y_{p+1} Y_{q} Y_{q+1} \\
	&\qquad\ \ \, + Y_{p} X_{p+1} X_{q} X_{q+1} - Y_{p} X_{p+1} Y_{q} Y_{q+1} + Y_{p} Y_{p+1} X_{q} Y_{q+1} + Y_{p} Y_{p+1} Y_{q} X_{q+1} ). \notag
\end{align}
The $ a_{p}^\dagger a_{q}^\dagger a_{q+1} a_{p+1} $ terms feature the same Pauli operators, but with a different sign pattern in the linear combination.

The 3- and 4-local terms are best handled separately. For the 3-local terms, the $ q $th index is free to take any value different from $ p $ and $ p+1 $. We measure these qubits in the computational basis, and the $ p $ and $ (p+1) $th qubits in the same fashion as in the 1-RDM case. Then to obtain all combinations (triples) $ (p,p',q) $, we have to swap all pairs $ (p,p') $ into, say, the qubit ordering $ (0,1) $. This allows a single ordering to account for $ (n-2) $ triples. There are $ \binom{n}{2}(n-2) $ total triples to permute into, and so we require $ \binom{n}{2} $ different swap circuits. To measure all 8 terms, we require 4 different Pauli bases, thus $ 4\binom{n}{2} $ different measurement circuits.  In practice we anticipate using the parallel transposition sort of the $1$-RDM measurement with an additional $n-1$ measurements at each swap circuit to account for each $(p,p', q)$ triple as $n-2$ $(p, p', q)$ triples can be acquired simultaneously by measuring $p, p'$ in either $X$ or $Y$ and all other qubits in $Z$.

The 4-local terms require us to swap the orbital orderings into 2-combinations of pairs in order to measure all general $ a_{p}^\dagger a_{p'}^\dagger a_{q} a_{q'} $ terms. First, the number of nearest-neighbor pairs in $ \{0,\ldots,n-1\} $ is $ n-1 $. Then we count how many 2-combinations of these pairs we can construct. Note that locality between 2-combinations is not a constraint, since $ p $ and $ q $ do not have to be local. However, they must be disjoint, otherwise we are double counting the 3-local terms. Each swap circuit can account for $ \binom{n-2}{2} $ different pairs of pairs;~thus, since there are $ \binom{n}{2}\binom{n-2}{2} $ unique $ a_{p}^\dagger a_{p'}^\dagger a_{q} a_{q'} $ terms, exactly $ \binom{n}{2} $ swap circuits are required. This statement is proved by a simple combinatorial argument, which we defer to the generalization in \cref{sec:k-generalization}.

The 16 Pauli operators cannot be measured in 16 bases whenever $n > 7$. In particular, from Eq.~\eqref{eq:general_tpdm} we observe that only the $XXXX$ and $YYYY$ geminals can be measured simultaneously for $n > 7$. Therefore, in order to read off the 4-local Pauli observables, we rely on quantum overlap tomography (QOT)~\cite{cotler2020quantum} to provide asymptotic bounds.  Each $(p,p',q,q')$ term requires measuring almost all $4$-qubit marginals.  More precisely, we need only the marginal elements corresponding to expectation values of Pauli operators composed of $X$ and $Y$ operators.  QOT provides a bound for $4$-qubit bound that has $\mathcal{O}(\log(n))$ scaling whenever the perfect hash family is known.  Though large sets of the $(n,4)$-perfect hash families are documented, there remain significant gaps in $n$.  

A more general procedure that does not achieve the same asymptotic bound of Ref.~\cite{cotler2020quantum} is the construction of a suboptimal perfect hash family by bootstrapping from the binary partitioning scheme described in Ref.~\cite{bonet2019nearly}.  Using this approach, the leading order complexity, which upper bounds the true scaling, in the number of measurement settings is defined as the function EQOT (explicit quantum overlap tomography):
\begin{equation}\label{eq:qot_extende_scaling}
\mathrm{EQOT}(k, n) = 3^{k}(k - 1)\sum_{m = 0}^{\lceil \log(n) \rceil - 1} m^{k - 2}.
\end{equation}
This gives the number of partitions to measure all $k$-qubit marginals, multiplied by all $3^k$ bases for each partition. In the case of the 4-local 2-RDM terms, however, we do not need to measure in any Pauli basis containing a $Z$, so we can straightforwardly reduce the base of this prefactor to $2$.  Altogether, we can upper bound the number of circuit configurations as
\begin{equation}
M_{T}(\text{2-RDM}) = 1 + 4\binom{n}{2} + \binom{n}{2} \mathrm{EQOT}(4, n) \l(\frac{2}{3}\r{)^4}
\end{equation}
circuits to measure the full fermionic 2-RDM with this approach.  Again we emphasize that, though this $ \mathrm{EQOT}(4, n) $ scaling is not optimal, it will work for any $n$.  In the following section we prove the partition scaling for the component of the $k$-RDM with $2k$ distinct indices, and thus the complexity, given an optimal construction of the swap circuits.

\subsection{\label{sec:k-generalization}A $ k $-RDM generalization}

Each $ k $-RDM observable decomposes into up to $ 4^k $ Pauli operators, thus is it impractical to write out such terms for general $ k $ by hand. Nonetheless, we can still obtain a scaling estimate on how many measurement circuits are required to reach all $ k $-RDM elements. The following argument also applies for the 4-local terms of the 2-RDM in the previous section. Consider the asymptotically dominant terms, $ a_{p_1}^\dagger \cdots a_{p_k}^\dagger a_{p_k+1} \cdots a_{p_{2k}} $ where all $ p_i\neq p_j $ for $ i\neq j $. There are $ \binom{n}{k}\binom{n-k}{k} $ such index combinations. The number of $ k $-combinations of disjoint nearest-neighbor pairs taken from $ \{(0,1),(1,2),\ldots,(n-2,n-1)\} $ is equivalent to counting how many unique sets of $ k $ nonconsecutive integers from $ \{0,\ldots,n-2\} $ exist. This is a classic ``stars and bars'' combinatorial problem and has solution $ \binom{(n-1)-k+1}{k}=\binom{n-k}{k} $. We can see this by a visual argument:~write down $ (n-1)-k $ spaces where the unchosen numbers will be placed in order. Then there are $ (n-1)-k-1 $ gaps in between the spaces, plus the $ 2 $ endpoints, where the chosen numbers can be placed, of which we choose $ k $. Thus we require $ \binom{n}{k} $ different swap circuits, hence $ \Omega\l[\binom{n}{k} 4^k\r] $ unique measurement circuits, to reach all $ k $-RDM elements. Similar to the 2-RDM case, this lower bound is in general not achievable, as we require (E)QOT to measure the Pauli operators in parallel, thus incurring polylogarithmic factors. We can upper bound this scaling by counting circuit repetitions for each $k$-RDM element partitioned by the number of unique indices.

\subsubsection{Upper bounds for $k = 3, 4$}
To derive an upper bound  for arbitrary $n$ in terms of measurement configurations for the 3-RDM and 4-RDM, we can follow the same procedure as the 2-RDM:~count circuits for measuring RDM terms after partitioning based on the number of unique indices in each RDM element.  The 3-RDM is partitioned into terms with 3, 4, 5, 6 different indices.  This can be checked by building a basis for the unique 3-RDM elements indexed by a tuple $(p,q,r)$ with $p < q < r$.  The terms with only 3 unique indices are analogous to the 2-index case for the 2-RDM---e.g., the Jordan--Wigner transformation of the three index term corresponds to diagonal 3-RDM elements and thus involves only Pauli-$Z$ operators, as follows. All three index terms are of the form
\begin{equation}
a_{p}^{\dagger}a_{q}^{\dagger}a_{r}^{\dagger}a_{r}a_{q}a_{p} = \frac{-Z_{p} - Z_{q} - Z_{r} + Z_{p}Z_{q} + Z_{p}Z_{r} + Z_{q}Z_{r} - Z_{p}Z_{q}Z_{r}}{8}
\end{equation}
and can be measured in one permutation of qubits (the identity permutation).

The terms with four unique indices are similar to the 3-index 2-RDM case.  For example, consider the real component of a 4-index 3-RDM term,
\begin{equation}
a_{p}^{\dagger}a_{q}^{\dagger}a_{r}^{\dagger}a_{r}a_{q}a_{p + 1} + \mathrm{h.c.} = \frac{1}{8}[X_{p}X_{p + 1} + Y_{p}Y_{p+1} + (X_{p}X_{p+1} + Y_{p}Y_{p+1})( Z_{q}Z_{r} -Z_{q} - Z_{r} )].
\end{equation}
These terms can be measured in $n/2$ circuits for the $XX$ and $YY$ parts, followed by $\binom{n}{2}$ circuits for the $XXZ$ and $YYZ$ terms, using the fact that $XX + YY$ commutes with $ZZ$, allowing us to use the measurement circuit from Ref.~\cite{arute2020hartree}.  At each of the $n/2$ circuit configurations, all $n-1$ pairs must account for all other $Z$ operators.  Thus in total, counting real and imaginary parts, we have a total of $4 \binom{n}{2}$ qubit permutations to measure all 4-index terms.

The 5-index terms contain one $Z$ term and can be measured with swap circuits analogous to the 4-index case in the 2-RDM.  Consider the example of the real-component of the 5-index 3-RDM term,
\begin{equation}
a_{p}^{\dagger}a_{q}^{\dagger}a_{r}^{\dagger}a_{r}a_{q+1}a_{p+1} + a_{p+1}^{\dagger}a_{q+1}^{\dagger}a_{r}^{\dagger}a_{r}a_{q}a_{p} = \frac{1}{16}\left( A_{pq} + Z_{r}A_{pq} \right),
\end{equation}
where we define
\begin{equation}
\begin{split}
A_{pq} &\equiv (X_{p}X_{p+1}X_{q}X_{q+1} + X_{p}X_{p+1}Y_{q}Y_{q+1} + X_{p}Y_{p+1}X_{q}Y_{q+1} - X_{p}Y_{p+1}Y_{q}X_{q+1} \\
&\quad + Y_{p}X_{p+1}Y_{q}X_{q+1} - Y_{p}X_{p+1}X_{q}Y_{q+1}  + Y_{p}Y_{p+1}X_{q}X_{q+1} + Y_{p}Y_{p+1}Y_{q}Y_{q+1} ),
\end{split}
\end{equation}
with $r \notin \{p, p+1, q, q+1\}$.  We can obtain an upper bound for the number of unique measurements settings as $\binom{n}{k}M$, where $M$ is the complexity of measuring $5$-qubit marginals, again using the technique of Ref.~\cite{bonet2019nearly}. The scaling for $M$ is given in Eq.~\eqref{eq:qot_extende_scaling} for $k=5$ and has a prefactor of $3^{5}$ to measure all $X$, $Y$, and $Z$ terms.  Finally, the 6-index term involves 64 separate terms, including the imaginary terms, which is measured by constructing the $\binom{n}{k}$ swap circuits and using EQOT for each permutation on the $6$-qubit marginal terms which make up the 3-RDM element.  
\begin{figure}[t]
    \centering
    \includegraphics[width=0.7\columnwidth]{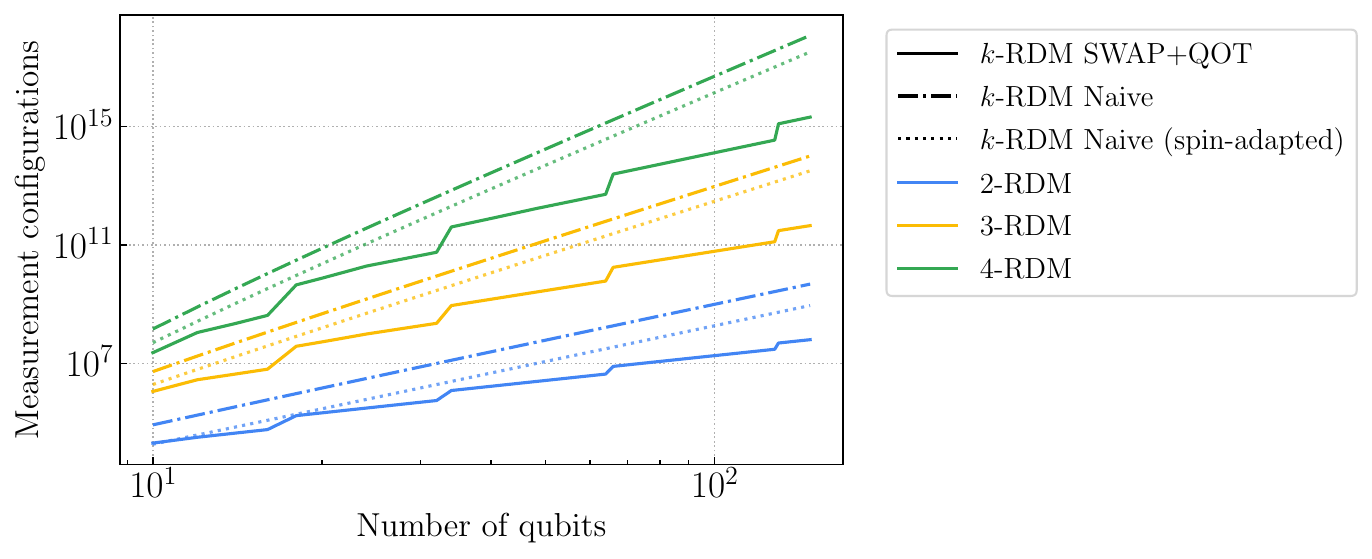}
    \caption{Scaling of the number of measurement configurations needed while measuring the $k$-RDM, for $k=2, 3, 4$ (blue, yellow, green), via the swap network protocol combined with quantum overlap tomography (solid lines) compared against the protocol where the unique upper triangle of the supermatrix representing the $k$-RDM is measured.  For the naive measurement strategy we scale the unique number  of terms in  the $k$-RDM by  $3^{2k}$ which is  obtained by counting products of $\sigma^{+},\sigma^{-},Z$. A constant-factor improvement is possible by taking advantage of $ S_z $-spin symmetry and only considering the spin-adapted blocks of the RDM supermatrix~\cite{rubin2018application}. The swap protocol upper bound is a quadratic improvement over the naive scaling. The solid curves corresponding to the SWAP+EQOT protocol are equivalent to the orange points in Fig.~\ref{fig:mainfig} of the main text.} 
    \label{fig:eqot_vs_naive}
\end{figure}

Overall, for the 3-RDM we can loosely bound the number of measurements as
\begin{equation}
M_{T}(\text{3-RDM}) = 1 + 4\binom{n}{2}  + \binom{n}{2}\mathrm{EQOT}(5)+ \binom{n}{3}\mathrm{EQOT}(6, n) \left(\frac{2}{3}\right)^{6},
\end{equation}
where each term corresponds to measuring the 3-, 4-, 5-, and 6-index terms of the 3-RDM, respectively. We note that this is an overestimate since Eq.~\eqref{eq:qot_extende_scaling} provides an upper bound to the $k$-qubit marginal measurement.  It is further loosened by the fact that we clearly do not need to measure \emph{all} $k$-qubit marginal terms.  

A similar accounting can be performed for the 4-RDM by breaking the unique 4-RDM elements into sets consisting of terms with 4, 5, 6, 7, and 8 unique indices.  This gives the upper bound on measurement configurations as
\begin{equation}
M_{T}(\text{4-RDM}) = 1 + 4\binom{n}{2} + \binom{n}{2}\mathrm{EQOT}(6, n) + \binom{n}{3}\mathrm{EQOT}(7, n) + \binom{n}{4}\mathrm{EQOT}(8, n) \left(\frac{2}{3}\right)^{8}.
\end{equation}

In \cref{fig:eqot_vs_naive} we plot the scaling of the swap network EQOT protocol against naively measuring the upper triangle of unique $k$-RDM elements in a supermatrix representation, as a function of the number of fermionic modes.

\section{\label{sec:numerics_appendix}Supplementary numerical calculations}

Here we provide some additional findings with our numerical studies. These are not essential to the primary results of the main text, but rather serve to explore some of the more subtler points of our partial tomography scheme.

\subsection{\label{sec:hyperparameter}Hyperparameter tuning}

As mentioned in the main text, we may control how many circuits to randomly generate by setting a hyperparameter $ r $ such that all $ K_r $ unitaries correspond to measurements of all observables at least $ r $ times each. Since increasing the sample size decreases the frequency of outlier events (i.e., some subset of observables being accounted for more often than the rest), we then expect that the value of $ K_r/r $ decrease as a function of $ r $. Indeed, this is a generic feature of randomization, as observed in the numerical results of the original work on classical shadows~\cite{huang2020predicting}. This effect is demonstrated, with the 2-RDM as an example, in \cref{fig:K_vs_r}. The behavior is consistent as expected.

\begin{figure}[h]
	\includegraphics[width=0.45\columnwidth]{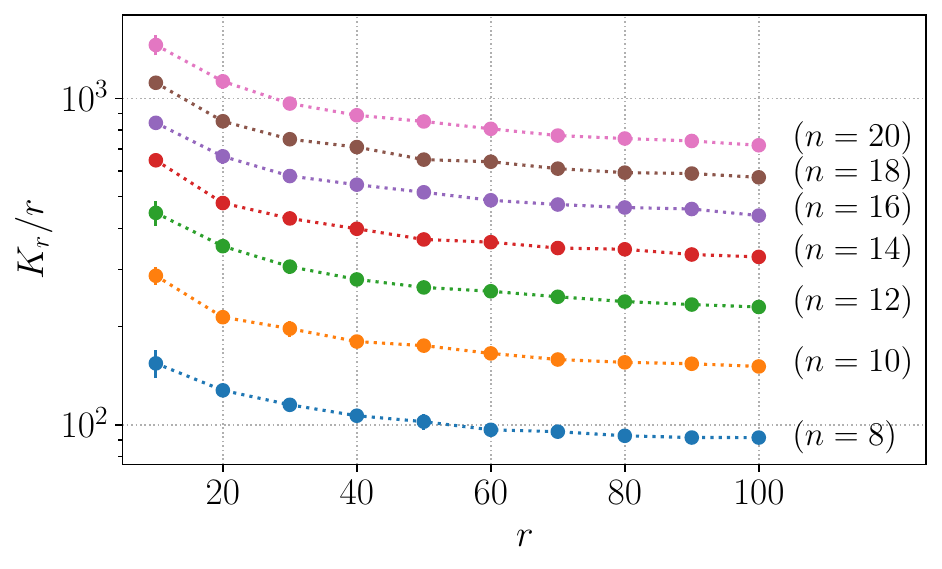}
	\caption{The relation between $ K_r/r $ and $ r $ for various numbers of modes under the $ \mathcal{U}_{\FGU} $ ensemble, with respect to covering $ 2 $-RDM observables. Since $ K_r $ is a random variable, we average over 10 randomly generated circuit collections for each $ r $ and indicate 1 standard deviation with uncertainty bars.}
	\label{fig:K_vs_r}
\end{figure}

\subsection{Realistic time estimates}

The main experimental difference between deterministic and randomized measurement schemes is the number of unique circuits one must run. In general, the number of random unitaries ($ K_r $) will be larger than the deterministic clique cover size ($ C $). Thus depending on the architecture, reprogramming the quantum device for each new circuit may incur a nontrivial overhead in the actual wall-clock time of the algorithm. Here, we show that for realistic experimental parameters, this consideration does not affect the results presented in the main text.

Suppose the quantum device can repeatedly sample a fixed circuit at a rate $ \fs $, but requires time $ \tl $ to load a new circuit. Then the total measurement time under the two paradigms are
\begin{align}
	T_{\mathrm{deter}} &= C \l( \frac{S}{\fs} + \tl \r), \label{eq:T_deter}\\
	T_{\mathrm{rand}} &= K_r \l( \frac{\lceil S/r \rceil}{\fs} + \tl \r), \label{eq:T_rand}
\end{align}
where we recall that $ S = \O(1/\varepsilon^2) $. Note that in the regime where $ S \gg \fs\tl $, we may directly compare $ C $ to $ K_r/r $, as in the main text. For hardware-dependent estimates, we take the specifications of the Google Sycamore chip as an example~\cite{arute2019quantum,arute2020hartree,arute2020quantum,arute2020observation}. The reported parameter values are $ \fs = 5 \times 10^3 $~Hz and $ \tl = 0.1 $~s~\cite{sung2020exploration}, and for \cref{fig:google_meas_times}, we set $ S = 2.5 \times 10^{5} $, in line with the number of shots taken to estimate 1-RDM elements in a recent Hartree--Fock experiment~\cite{arute2020hartree}. We observe no qualitative differences from the results of the main text.

\begin{figure}
	\includegraphics[width=\textwidth]{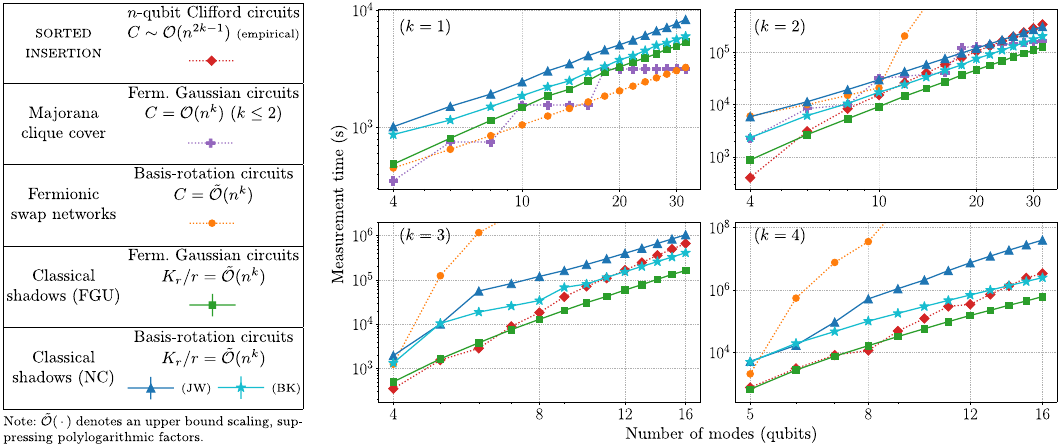}
	\caption{Measurement times for estimating $ k $-RDMs, calculated with \cref{eq:T_deter,eq:T_rand}, under the reported device parameters of the Google Sycamore chip~\cite{sung2020exploration} and with $ S = 2.5 \times 10^5 $. The underlying data is that of \cref{fig:mainfig} in the main text. For convenience, we reproduce the legend of the main text here.}
	\label{fig:google_meas_times}
\end{figure}

One may also study the performance of the different methods as a function of the target accuracy. In \cref{fig:time_scaling_with_S}, we show how $ T_{\mathrm{rand}} $ scales with $ S $ for the 2-RDM, using a few values of $ n $ as illustrative examples. Though the threshold at which randomization begins to outperform the other methods varies depending on $ n $ and $ k $, it typically lies below $ {\sim} \, 10^5 $, which is well under typical sampling requirements. Note that the $ S \gg \fs\tl $ regime corresponds to when $ T_{\mathrm{rand}} $ scales linearly with $ S $.

\begin{figure}
    \includegraphics[width=\textwidth]{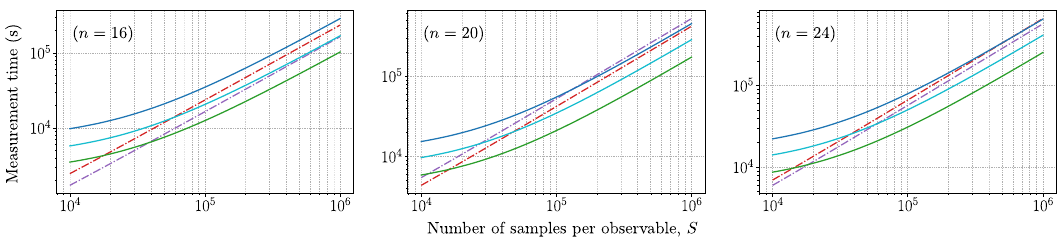}
    \caption{Examples of how the measurement times for our randomized schemes scale with the level of precision, as given by \cref{eq:T_rand}. We use the same device parameters here as in \cref{fig:google_meas_times}. For comparison, we also plot the linear scaling of the prior deterministic strategies (excluding the swap network method), which gives an indication for the values of $ S $ beyond which we obtain an advantage with our methods under this time-cost model. The colors correspond to those in the legend of \cref{fig:google_meas_times}.}
    \label{fig:time_scaling_with_S}
\end{figure}

\subsection{\label{sec:hamiltonian_appendix}Hamiltonian averaging}

In the context of estimating a single observable, whose expectation value is a linear combination of RDM elements, the number of circuit repetitions required is more properly determined by taking into account the coefficients of the terms and covariances between simultaneously measured terms~\cite{wecker2015progress,mcclean2016theory,rubin2018application}. This can be directly calculated from the (single-shot) variance of the corresponding estimator. Here we provide some preliminary numerical calculations in this context, with respect to $ \mathcal{U}_{\FGU} $. Although our uniformly distributed ensemble is not tailored for Hamiltonian averaging, these calculations provide a benchmark for potential improvement.

Without loss of generality, consider a traceless fermionic $ k $-body Hamiltonian
\begin{equation}
	H = \sum_{j=1}^k \sum_{\bm{\mu} \in \comb{2n}{2j}} h_{\bm{\mu}} \Gamma_{\bm{\mu}}, \quad  h_{\bm{\mu}} \in \R.
\end{equation}
For the numerical studies presented here, we consider a sample of molecular Hamiltonians (hence $ k = 2 $) at equilibrium nuclear geometry, obtained through OpenFermion~\cite{openfermion} interfaced with the Psi4 electronic structure package~\cite{psi4}. We used a minimal STO-3G orbital basis set to generate these Hamiltonians, except for the $ \mathrm{H}_2 $ molecule, which was represented in the 6-31G basis.

In Table~\ref{tab:H_variances} we compare our classical shadows (CS) $ \mathcal{U}_{\FGU} $ ensemble against two prominent measurement schemes for electronic-structure Hamiltonians:~a strategy based on a low-rank factorization of the coefficient tensor, termed basis-rotation grouping (BRG)~\cite{huggins2019efficient}, and a locally biased adaptation on classical shadows (LBCS)~\cite{hadfield2020measurements}. For reference, we also report the variance under standard classical shadows with Pauli measurements~\cite{huang2020predicting}. The expressions for the variances in terms of the Hamiltonian terms and a reference state $ \rho $ (taken here to be the ground state) are provided in their respective references. While we provide a state-independent upper bound in \cref{sec:arb_observable_appendix}, the exact variance expression for our $ \mathcal{U}_\FGU $ ensemble on arbitrary observables has been derived in Refs.~\cite{wan2022matchgate,ogorman2022fermionic}.

Note that, in order to compare fairly between the deterministic (BRG) and randomized methods (CS, LBCS), we reframe the deterministic measurement scheduling such that an equivalent variance quantity may be computed. This principle was also used in the numerical comparisons of Ref.~\cite{hadfield2020measurements}, which we generalize here. We decompose the target Hamiltonian as
\begin{equation}
    H = \sum_{\ell=1}^L O_\ell,
\end{equation}
where each $ \tr(O_\ell \rho) $ may be estimated by a single measurement setting (as defined by the given strategy). The optimal distribution of measurements then allocates a fraction
\begin{equation}\label{eq:optimal_fraction}
    p_\ell \coloneqq \frac{\sqrt{\V_{\rho}[O_\ell]}}{\sum_{j=1}^L \sqrt{\V_{\rho}[O_j]}}
\end{equation}
of the total measurement budget to the $ \ell $th setting~\cite{rubin2018application}. The variance here is simply the quantum-mechanical operator variance,
\begin{equation}
    \V_{\rho}[O_\ell] = \tr(O_\ell^2 \rho) - \tr(O_\ell \rho)^2.
\end{equation}
Recognizing $ \{p_\ell\}_\ell $ as a collection of positive numbers which sum to unity, we may recast the deterministic strategy into the language of randomization, where the unbiased estimator is given by
\begin{equation}
    \E_{\ell, \rho}\l[ \frac{1}{p_\ell} O_\ell \r] = \E_{\ell}\l[ \frac{1}{p_\ell} \tr(O_\ell \rho) \r] = \tr(H \rho).
\end{equation}
The variance of this estimator is therefore
\begin{equation}\label{eq:variance_general}
    \begin{split}
    \V_{\ell, \rho}\l[ \frac{1}{p_\ell} O_\ell \r] &= \E_{\ell, \rho}\l[ \frac{1}{p_\ell^2} O_\ell^2 \r] - \E_{\ell, \rho}\l[ \frac{1}{p_\ell} O_\ell \r{]^2}\\
    &= \sum_{\ell=1}^L \frac{1}{p_\ell} \tr(O_\ell^2 \rho) - \tr(H \rho)^2.
    \end{split}
\end{equation}
Note that this analysis does not actually require a randomization the deterministic strategy, but merely normalizes the measurement allocations so as to produce an equivalent figure of merit.

\begin{table}
    \caption{Variances of Hamiltonian averaging estimators under various strategies, reported in units of $ \mathrm{Ha}^2 $. The expressions for the CS (Pauli) and LBCS variances may be found in their original references, although the fundamental formula for all table entries is given by \cref{eq:variance_general}, where the decomposition into measurable terms $ O_\ell $ and probabilities $ p_\ell $ are determined by the particular method. For CS (FGU), the explicit expression is provided in Refs.~\cite{wan2022matchgate,ogorman2022fermionic}. The reference state used here is the exact ground state.}
    \begin{tabular}{c c c c c}
    \toprule
    {} & \multicolumn{4}{c}{Methods}\\
    \cmidrule{2-5}
    Molecule (qubits) & CS (Pauli)~\cite{huang2020predicting} & LBCS~\cite{hadfield2020measurements} & BRG~\cite{huggins2019efficient} & CS (FGU)\\
    \midrule
    $ \text{H}_2 $ (8) & 51.4 & 17.5 & 22.6 & 69.6\\
    $ \text{LiH} $ (12) & 266 & 14.8 & 7.0 & 155\\
    $ \text{BeH}_2 $ (14) & 1670 & 67.6 & 68.3 & 586\\
    $ \text{H}_2\text{O} $ (14) & 2840 & 257 & 6559 & 8440\\
    $ \text{NH}_3 $ (16) & 14400 & 353 & 3288 & 5846\\
    \bottomrule
    \end{tabular}
	\label{tab:H_variances}
\end{table}

While \cref{eq:optimal_fraction} provides the optimal distribution of measurements, one may use any distribution in its place. In particular, because $ \rho $ is unknown, a state-independent approximation to $ \{p_\ell\}_\ell $ is often a more practical option;~tighter bounds may be obtained by a classically tractable approximation to the true, unknown state. For simplicity, in \cref{tab:H_variances} we take $ \rho $ as the exact ground state within the model chemistry. We note that Ref.~\cite{huggins2019efficient} showed evidence that the discrepancy between using the exact ground state and the state obtained from a configuration interaction with single and double excitations (CISD) calculation is negligible in this context.

Our main takeaway from \cref{tab:H_variances} is that, similar to how Pauli measurements can be dramatically improved via solving an optimization problem targeted specifically at minimizing this variance~\cite{hadfield2020measurements}, our classical shadows method similarly has a large room for improvement. Since only LBCS performs such an optimization, it is perhaps not too surprising that it is the most efficient approach here, despite only employing Pauli measurements. Encouragingly, such a biasing (or similar optimization-based extensions) may also be applied to our classical shadows ensemble in principle.
It should be noted that these results assume the noiseless case;~for instance, the BRG strategy additionally offers resilience to device errors and the ability to postselect on particle number, which have the effect of reducing noise-induced contributions to the variance~\cite{huggins2019efficient}.

\bibliography{references}

\end{document}